\documentclass[10pt]{acmart}%\settopmatter{printfolios=true}
\input{preamble}
\newenvironment{ottdefnblock}[3][]{ \framebox{\mbox{#2}} \quad #3 \\[0pt]}{}

\newcommand{\ottnt}[1]{\mathit{#1}}
\newcommand{\ottmv}[1]{\mathit{#1}}

\newcommand{\ottsym}[1]{#1}

\title{Structural Operational Semantics for Control Flow Graph Machines}
\subtitle{Draft Version for arXiv Submission}

 \author {Dmitri Garbuzov,   William Mansky,  Christine Rizkallah, Steve Zdancewic}
\email{University of Pennsylvania \\ \{dmitri, wmansky, criz, stevez\}@cis.upenn.edu}

\begin{document}
\begin{abstract}
Compilers use control flow graph (CFG) representations of low-level
programs because they are suited to program analysis and
optimizations. However, formalizing the behavior and metatheory of CFG
programs is non-trivial: CFG programs don't compose well, their
semantics depends on auxiliary state, and, as a consequence, they do
not enjoy a simple equational theory that can be used for reasoning
about the correctness of program transformations.
Lambda-calculus-based intermediate representations, in contrast, have
well-understood operational semantics and metatheory, including rich
equational theories, all of which makes them amenable to formal
verification.

This paper establishes a tight equivalence between (a variant of)
Levy's call-by-push-value (CBPV) calculus and a control flow graph
machine whose instructions are in static single assignment (SSA) form.
The correspondence is made precise via a series of abstract machines
that align the transitions of the structural operational semantics of
the CBPV language with the computation steps of the SSA form.

The target machine, which is derived from the CBPV language,
accurately captures the execution model of control flow graphs,
including direct jumps, mutually recursive code blocks, and
multi-argument function calls, and the closure-free subset is similar
to the SSA intermediate representations found in modern compilers such
as LLVM and GCC.
The definitions of all the language/abstract machine semantics and the
theorems relating them are fully verified in Coq.

% The key technical result proves that the high-level and low-level
% operational semantics are different representations of the
% same transition system. This means that a wide variety of
% operational equivalences defined on CBPV terms can be transported to
% the abstract machines while preserving reasoning principles (for
% instance about execution costs). As such, this paper is a step towards
% translating results and techniques between lambda-calculus-based and
% control-flow-graph-based formalisms.

%% Local Variables:
%% fill-column: 70
%% eval: (auto-fill-mode)
%% eval: (flyspell-mode)
%% eval: (setq-local auto-hscroll-mode nil)
%% eval: (setq-local sentence-end-double-space nil)
%% End:

\end{abstract}

\maketitle

% \category{CR-number}{subcategory}{third-level}

% % general terms are not compulsory anymore,
% % you may leave them out
% \terms
% term1, term2

% \keywords
% keyword1, keyword2

%%%%%%%%%%%%%%%%%%
%% PAPER STARTS %%
%%%%%%%%%%%%%%%%%%

%% To include a new file ``filename.mng.tex'': 
%% \include{temp/filename}

\section{Introduction}
\label{sec:introduction}

\paragraph{Control flow graphs} Classical compiler theory, data-flow analysis, abstract
interpretation, and optimization are all traditionally formulated in
terms of \textit{control-flow-graph} (CFG) machines~\cite{Muchnick, Dragonbook}.
CFG-based machines serve as abstract, intuitive models of processor
execution with features that
are suitable for code generation and reasoning about the cost of
execution on typical hardware---register files, program counters,
and direct and indirect jumps to labeled blocks of instructions. These properties make CFG representations
suitable for some low-level optimizations, analyses, and
transformations, so they are widely used in practice.

However, if one wants to create machine-checked \textit{proofs} (as in
CompCert~\cite{compcert}) of properties about program transformations,
such as the correctness of optimization passes, then CFG
representations leave a lot to be desired. The same details that make
CFG machines good models of processor execution must be explicitly
ignored to define suitable notions of program equivalence.
Furthermore, unlike term models, CFGs do not have a single obvious
notion of composition, making contextual reasoning and reasoning about
transformations of CFG programs difficult.
%%%% SSA is not compositional: Adding a new jump that targets the label
%%%% \kwl{f} from some other block requires changing the \kwl{phi}
%%%% instructions at the start of \kwl{f}.}
%For
% example, proving the correctness of even a ``simple'' CFG optimization
% that control-flow blocks (to eliminate a direct jump), code inlining,

\textit{Static single assignment}~\cite{CytronFRWZ91} (SSA) form is a
restricted subset of CFG programs where each pseudo-register is
assigned by a single instruction. A form of conditional parallel move
is supported by special ``phi instructions'' that are inserted at
control-flow merge points. SSA-form CFGs were introduced to improve
the performance of common dataflow analyses and have become ubiquitous
in compiler intermediate representations (IRs). While it was recognized early on that there is a
syntactic similarity between SSA and functional IRs like CPS and ANF,
the relationship between their semantics has never been made precise.
Work formalizing the semantics of SSA intermediate representations,
including Vellvm and CompCertSSA, uses CFG semantics.

%% One (partial)
%% solution to the issues above is to impose additional invariants on CFG
%% programs that make them easier to work with. A compelling choice for
%% those invariants is to restrict CFG code to \textit{static single
%%   assignment}~\cite{CytronFRWZ91} (SSA) form, which ensures that local
%% variables are well scoped and eliminates unnecessary state in favor of
%% immutable bindings for local variables. SSA yields an efficient data
%% structure for working with code and is conducive to many dataflow
%% analyses and optimizations, so it has been adopted by prominent
%% compiler implementations like GCC~\cite{gcc} and LLVM~\cite{llvm}.
%% \steve{Why is this a partial solution?  Also do we want to mention
%%   ANF/CPS/etc. briefly here too?}

\textbf{\textit{Structural operational semantics}}\quad Structural operational
semantics~\cite{sos}  (SOS) was developed specifically
to address the shortcomings of working with such low-level machine
models for reasoning about the operational behavior of programming
languages.
Defining operational semantics based on the structure of program terms
and using meta-level operations like substitution leads naturally to
compositional definitions of program equivalences and compositional
analyses like type systems that enforce scoping or other
well-formedness invariants. Lambda-calculus-based representations of
programs have a well-understood metatheory that offers a wealth of
reasoning techniques~\cite{ATAPL,pitts}. They are also
well-suited for machine formalization.
%% \cite{metatheoryforthemasses}.

However, while there is a large body of work on abstract machines for
languages based on the lambda calculus and the precise relationships
between different styles of operational semantics (see the discussion
on related work in Section~\ref{sec:discussion}), these techniques
have not yet been applied in the context of formalizing the semantics
of control flow graphs.

\textbf{\textit{Contributions}}\quad 
Motivated by the need to develop improved techniques for reasoning
about low-level code in the context of verified compilers, this paper
demonstrates how to reconcile the desiderata above.
Our main contribution is a proof of \textit{machine equivalence}
between (a variant of) Levy's call-by-push-value (CBPV) lambda
calculus~\cite{levy} and a virtual machine that operates on control
flow graphs of instructions in SSA form. 

On the one hand, CBPV is a variant of the lambda calculus with a
well-understood SOS, metatheory, and a robust set of reasoning
principles that are amenable to formal verification. On the other
hand, the CFG machine we describe accurately captures many features of
the ``close-to-the-metal'' execution model, including direct jumps,
mutually recursive code blocks, and multi-argument function calls. CFG
programs in the image of our translation enjoy SSA-like invariants,
and, as a consequence, the closure-free fragment is similar to IRs in
modern compilers such as LLVM and GCC. Additionally, because it is
derived from CBPV, which is a higher-order language, our CFG formalism
smoothly extends to higher-order features.

This result shows that the CBPV and CFG semantics are different
representations of the \textit{same} transition system. This is
stronger than simply a compiler-correctness result: the CFG
machine is correct with respect to \textit{any} choice of
observations of the CBPV SOS, not just a particular program
equivalence, e.g. bisimilarity of traces. Furthermore,
every well-formed (SSA) control flow graph corresponds to some
CBPV term. SSA control flow graphs are then just another way of
efficiently reducing lambda terms. It is possible to take a (possibly
open) CBPV term, produce the corresponding CFG, execute an arbitrary
number of steps of the CFG machine, and recover the residual lambda
term. Techniques for establishing program equivalence, symbolic
execution, and type systems for lambda calculi can be transported to
the CFG machine.

% Coq proofs verified compiler for functional language.
% This state of affairs is somewhat surprising given that Appel declared
% some time ago that ``SSA \textit{is} Functional
% Programming''~\cite{appelssa} and argued that SSA form is essentially
% equivalent to the CPS subset of lambda calculus. 

% Contextual rewrites: equivalences that hold in open context (mention
% advantages over CompCert later)

Although the focus of this paper is on machine-checked proofs of the
equivalence claims mentioned above, these results have potential
applications in compiler implementations. Compilers for functional
languages, which often already used restricted forms of
lambda calculus (such as CPS or ANF), could instead use our CBPV
representation, which solves some issues with those IRs.
%% (see the discussion in Section~\ref{sec:todo})\steve{TODO: where}. 
One could
also imagine implementing a compiler for an imperative language like C
that, rather than targeting an SSA-form CFG IR, instead targets the
first-order subset of our CBPV lambda calculus. We leave such
investigations to future work.

The rest of the paper proceeds as follows. Section \ref{sec:problem}
demonstrates some key features of a CFG machine and foreshadows the
results of the paper by example.
Section \ref{sec:cbpv} introduces the
call-by-push-value calculus, its SOS semantics, and a corresponding
simple abstract machine.
Section \ref{sec:proof} sketches the high-level structure of the main result
and explains our formulation of machine equivalence.
The meat of the paper follows in sections \ref{sec:peak}, \ref{sec:pek}, and \ref{sec:cfg}. Inspired by
the techniques of Ager,\etal~\cite{danvymachines}, we derive the CFG machine via
a series of abstract machines by systematically teasing apart the
static content (\ie{} the control flow graph of instructions) from the
dynamic content (\ie{} the environment and call stack) of a CBPV
term.
Section~\ref{sec:applications} shows how we can put our machine equivalence theorem to work, using it to relate equations on CBPV terms to code transformations. This gives us a technique for verifying compiler optimizations using the equational theory of CBPV.
 We conclude by comparing
our result to existing work in Section~\ref{sec:discussion} and summarizing our results in Section~\ref{sec:conclusion}.
% In particular, our results apply to the entire \textit{pure}
% subset of SSA but omit \textit{effects}; we briefly describe how the
% strong machine equivalence can be leveraged to incorporate reasoning
% about arbitrary observational equivalences, but we leave formalization
% to future work.\steve{CHECK THIS}

% (including
% higher-order functions), not just the intraprocedural fragment.

The definitions of all the language/abstract machine semantics and the
theorems relating them are fully verified in Coq. The formalization is available online~\cite{formalization}.

% Our main contribution is to extend the method of deriving machine
% models presented in

% We prove a strong form of machine equivalence between a structural
% operational semantics for a call-by-push-value lambda calculus and a
% derived control flow graph machine.

% That is, the two forms of operational semantics define the same
% transition system. 

% The resulting CFG machine captures key features of SSA form.

%% Rest of the paper

%% Local Variables:
%% fill-column: 70
%% eval: (auto-fill-mode)
%% eval: (flyspell-mode)
%% eval: (setq-local auto-hscroll-mode nil)
%% eval: (setq-local sentence-end-double-space nil)
%% End:

\section{Background: Abstract Machines and SSA}
\label{sec:problem}

\steve{This tutorial about environment machines is good, but we might
  be able to make its points more concretely by giving examples of
  what we mean.}
This section outlines some of the differences between control flow
graph formalisms and abstract machines for languages, like lambda
calculus, with lexical scope. We assume some familiarity with control
flow graphs and their interpretation as found in a standard compiler
textbook~\cite{Muchnick, Dragonbook} as well as execution models for
languages based on the lambda calculus.
% We
% conclude with an overview of SSA-structured CFGs code by way of
% example and foreshadow our technical developments by showing the same
% program in both CFG and CBPV representations.

Abstract machines are a form of operational semantics in which
auxiliary symbolic structures are used to explicate intensional
aspects of execution. Compared with structural operational semantics
using a substitution meta-operation, the transition rules of an
abstract machine are more suited for implementation on
typical hardware. Accordingly, they are often used as the basis for
efficient interpreters or code generation in a compiler.

% \steve{Abstract machines only make sense for ``high-level'' languages, no?}

%
% From this perspective, the
% control flow graph machine we will derive is not a separate language
% with a useful notion of program equivalence on CFGs, but an efficient
% representation of the evaluator defined by the structural operational
% semantics.

%% TODO so this is how we eliminate the need for reasoning about CFGs?
%% TODO forward reference to precise definition of machine equivalence

\subsection{Environment Machines}

Though a wide variety of abstract machines have been devised for
languages based on the lambda calculus, a large class can be
characterized as \textit{environment machines}. In an environment
machine, the substitutions carried out in the SOS are represented by a
data structure that is used to look up bindings when the machine
encounters a variable. Machine states are easily mapped back to those
of the SOS by carrying out the delayed substitution.

One aspect of execution on a typical CPU captured by CFG formalisms is
the distinction between direct and indirect jumps. Program locations
are addressable by \textit{labels} and control flow transfers such as
jumps or function calls can target literal or computed labels.
Environment machines make a similar distinction: a function called
through a variable in the source program is executed by first
performing a lookup (an indirect jump), whereas an application of a syntactic lambda
skips the indirection through the environment (a direct jump). The obvious drawback to
this scheme is that it is impossible for a single function in the
source program to be the target of more than one direct jump;
functions that are control flow merge points must be referenced
through the environment.

There is another important distinction between calling functions
through a variable compared to a literal lambda in environment
machines: the latter requires allocating fewer closures. To apply a
syntactic lambda in the source program, the environment of the machine
at the point that control reaches the application is used. Bindings in
the environment, however, are uniformly represented as closures, even
if they are used for control flow that is essentially first-order.
Consider the program (in psuedo-ML):

\begin{center}
\begin{tabular}{c}
\begin{lstlisting}[mathescape=true]
let f = $\lambda$ x. M in 
. . .
if g then f c else f c'  
\end{lstlisting}
\end{tabular}
\end{center}

\steve{Can we show the steps of evaluating this example?}
Here, an environment machine would allocate a closure for $f$ when
control proceeds under the binder, saving the environment and
restoring it when $f$ is applied. This seems wasteful, considering
that the machine's environment can only be extended when control
reaches the application of $f$. Control-flow-graph machines, on the
other hand, have an execution model where environments do not have
to be saved and restored. A similar issue occurs simply with
expressions bound by \texttt{let}. When control reaches \texttt{let x
  = M in N} in an environment machine, the machine proceeds by
executing \texttt{M}, pushing \texttt{N} onto the stack. Even though
\texttt{M} might only extend the environment with additional bindings,
the current environment is also saved. To avoid these expensive
operations, compilers for functional languages use CPS or ANF
forms that restrict the \texttt{let} operators to binding only
values and primitives. This results in a simpler execution model, but
requires global rewriting of the source program \cite{cek}.
\steve{Reminder: we need to compare against Peyton Jones' ``join
  points'' ICFP paper somewhere, maybe not here.}

Finally, environment machines for the lambda calculus require
additional machinery to support curried function application. In
call-by-value languages lambdas serve a dual role in evaluating terms:
depending on context, they indicate suspending a computation and
returning to the calling context or creating a binding in the
environment. The CEK machine, for example, uses two types of frames to
indicate if execution is proceeding in function or argument position.

Syntactic restrictions like CPS and ANF explicitly sequence each
application with \texttt{let}, but extend the source with
multi-argument functions and n-ary application. This extension
better reflects the capabilities of the underlying machine, but
complicates the associated equational theory. ANF, in particular, is
not even closed under the usual beta reduction. \cite{kennedy07}
Rather than using a restricted subset of our source language to
enable a simple execution model, we allow the abstract machine to use
information from the context. Let-normalization in the sense of ANF is
performed during partial evaluation of the abstract machine rules to
obtain an instruction set.

In this work, we derive an execution model for a lambda calculus that
addresses the above issues. The need for currying machinery or
extensions for n-ary application is addressed by using Levy's
call-by-push-value (CBPV)~\cite{levy}. Direct jumps are modeled by treating
\textit{computed} bindings like those of lambdas and sequencing
computations separately from \textit{known} bindings indicating
sharing. Efficient use of environments is achieved using a novel
representation of environments. Partially evaluating the resulting
abstract machine produces SSA form control flow graphs for first-order
programs, and a higher-order extension of SSA in general.

% \subsection{Static Single Assignment Form}
% \textit{Static single assignment form} (SSA) is a popular family of
% compiler intermediate representations, used in both leading C
% implementations~\cite{llvm,gcc} and compilers for functional
% languages~\cite{MLton}. It was introduced primarily to increase the
% efficiency of common dataflow-based
% optimizations~\cite{CytronFRWZ91}, and is both an efficient
% datastructure manipulated by the optimizer and a programming language
% in its own right. There are many different variants of SSA IRs, and so
% there is no canonical syntax or operational model, but these
% various IRs all share some features: SSA code is organized as a
% control flow graph of instructions with the constraint that each
% variable defined by an instruction is assigned exactly once,
% statically. The representation has no explicit nested scope, and blocks can refer
% to any label or identifier. However, to be well-formed, each use of an
% identifier must be dominated by its unique definition in the CFG.

\lstdefinelanguage{ANF}{%
  morekeywords={%
    if0,then,else,let,rec,in,and % keywords go here
  },%
  morekeywords={[2]int},   % types go here
  otherkeywords={:,.,-,+,=}, % operators go here
  literate={% replace strings with symbols
    {lambda}{{$\boldsymbol{\lambda}$}}{1}
  },
  basicstyle={\sffamily},
  keywordstyle={\bfseries},
  keywordstyle={[2]\itshape}, % style for types
%  keepspaces,
  mathescape % optional
}[keywords,comments,strings]%
\lstdefinelanguage{SSA}{%
  morekeywords={%
    function,lbl,if0,sub,add,ret,jmp,else,phi % keywords go here
  },%
  literate={% replace strings with symbols
    {<-}{{$\leftarrow$}}{1}
  },
  otherkeywords={:,@}, % operators go here
  basicstyle={\sffamily},
  keywordstyle={\bfseries},
  keywordstyle={[2]\itshape}, % style for types
  keepspaces,
  mathescape % optional
}[keywords,comments,strings]%
\lstdefinelanguage{CFG}{%
  morekeywords={%
    POP,MOV,TAIL,IF0,RET,SUB,TAIL,ADD % keywords go here
  },%
  morecomment=[l]{//},
  otherkeywords={:,@}, % operators go here
  basicstyle={\footnotesize\ttfamily},
  keywordstyle={\bfseries},
  keywordstyle={[2]\itshape}, % style for types
  commentstyle={\rmfamily\itshape},
  keepspaces,
  mathescape % optional
}[keywords,comments,strings]%

\begin{figure}[tbp]
% \begin{minipage}[t]{.30\textwidth}
% \small
% \begin{lstlisting}[language=SSA]
% function power(n, m):
% e:
%   jmp f
% f:
%   m$_1$ <- phi(e:m, g:m$_k$)
%   a$_1$ <- phi(e:1, g:a$_2$)
%   if0 m$_1$ jmp f$_0$ else f$_1$
% f$_0$:
%   ret a$_1$
% f$_1$:
%   m$_2$ <- sub m$_1$ 1
%   jmp g
% g:
%   m$_3$ <- phi(f$_1$:a$_1$, g$_0$:m$_4$)
%   a$_2$ <- phi(f$_1$:0,$_{\;\,}$ g$_0$:a$_3$)
%   m$_k$ <- phi(f$_1$:m$_2$, g$_0$:m$_k$)
%   if0 m$_3$ jmp f else g$_0$
% g$_0$:
%   m$_4$ <- sub m$_3$ 1
%   a$_3$ <- add a$_2$ n
%   jmp g
% \end{lstlisting}
% \end{minipage}
\hspace{-1in}
\begin{minipage}[t]{0.35\textwidth}
\begin{lstlisting}[language=CFG]
let rec
  mult n x a =
    if x = 0 then 0 else
    let y = x - 1 in
      if y = 0 then a else
      let b = a + n in
      mult n y b
in 
  mult
\end{lstlisting}
\end{minipage}
\begin{minipage}[t]{0.20\textwidth}
\begin{lstlisting}[language=CFG]
0:  RET @1

1:  n = POP [2]
2:  x = POP [3]
3:  a = POP [4]
4:  CBR x [8] [5]

5:  y = SUB x 1 [6]
6:  CBR y [7] [9]
7:  RET a

8:  RET 0

9:  b = ADD a n [10]
10: TAIL @1 b y n
\end{lstlisting}
\end{minipage}
\qquad
\begin{minipage}[t]{.20\textwidth}
\[
\begin{array}[t]{l}
 \kw{letrec} \\
\qquad  mult =  \lambda  n.  \lambda  x.  \lambda  a. \\
\qquad  \kw{if0} \ x\ ( \kw{prd} \ 0) \\
\qquad \qquad (x - 1\  \kw{to} \ y\  \kw{in}   \\
\qquad \qquad \qquad ( \kw{if0} \ y\ ( \kw{prd} \ a)  \\
\qquad \qquad \qquad \qquad (a + n\  \kw{to} \ b\  \kw{in}  \\
\qquad \qquad \qquad \qquad \ n  \mathbin{\cdot}  y  \mathbin{\cdot}  b  \mathbin{\cdot}  mult ))) \\
 \kw{in} \\
\qquad mult\\
\end{array}
\]
\end{minipage}

% Example mult_body :=
% (* \n = 2  \ x= 1 \ acc = 0.*)
% Lam (Lam (Lam
%             (* if0 x then ... *)
%             (IfZ (Var 1) (Prd (Nat 0))
%                  (Seq (Aop OMinus (Var 1) (Nat 1))
%                       (IfZ (Var 0) (Prd (Var 1))
%                            (Seq (Aop OPlus (Var 1) (Var 3))
%                                 (App (Var 4) (App (Var 1) (App (Var 0) (Force (Var 5))))))))))). 

% Definition mult := Rec (Prd (Var 0)) [mult_body].

\caption{A sample program $mult$ represented as a control-flow-graph
  machine program as OCaml code (left) in  SSA form (middle) and as a call-by-push-value term (right).}
\label{fig:ssacfg}
\end{figure}

\subsection{Control-flow-graphs vs. Call-by-Push-Value}

To convey the intuition behind our machine correspondence, this
section gives and example program written three ways. The left part of
Figure~\ref{fig:ssacfg} shows a (silly) program $mult$ such that
$mult n m a$ computes $n * (m-1) + a$ by iterated addition,
represented as an OCaml program.  The middle part of the figure shows
the same code using the control-flow-graph virtual machine language
that we derive in this paper.  The right-hand side of the figure shows
the same program as a CBPV term.

A CFG program is represented by a mapping from program points
(numbers running down the left-hand side of the code) to instructions.
As in an SSA representation, each CFG instruction binds a value
to a local variable. So, for instance, the program fragment 
\code{0: n  = POP [2]}
pops the top element of the stack and assigns it to the local
identifier \code{n}. The  CFG
representation does not have implicit ``fall through'' for instruction
sequences. The numbers, like \code{[2]}, written in square brackets at
the end of an instruction indicate the program point to which control
should pass after the instruction executes. A conditional branch like
\code{4: CBR x [7] [9]} selects one of two successor program points
based on the value of its first operand.
(As we will see later, in general CFG programs may treat labels as
first-class values and can use indirect jumps.)

% In our CFG representation, \kw{phi} nodes are replaced by \kwl{POP}
% instructions and the corresponding \kw{jmp}s are replaced with
% argument-taking \kwl{TAIL} call instructions. For example, the tail
% call instruction \kwl{13: TAIL @3 m$_k$ a$_2$} transfers control to
% the program point \kwl{3}---the \kwl{@} marker indicates that the
% program point is being used as a code label that can be called. This
% particular tail call corresponds to the \kw{g} entries of the
% two \kw{phi} nodes at label \kw{f} in the SSA version.

The right side of Figure~\ref{fig:ssacfg} shows the same program yet
again, but this time represented using a call-by-push-value (CBPV) term. We
have chosen the variable names so that they line up precisely with the
CFG representation. We will explain the semantics of CBPV programs and
its correspondence with the CFG machine in more detail below.
For now, it is enough to observe that each $\lambda$-abstraction
corresponds to a \code{POP}. Application is written in reverse order
compared to usual lambda calculus and CBPV also has a sequencing operation $\ottnt{M} \, \kw{to} \, \ottmv{x} \, \kw{in} \, \ottnt{N}$.

\section {CBPV and its CEK Machine}
\label{sec:cbpv}
% \begin{itemize}
% \item history
% \item motivation \\
%   * no cbn/cbv split, can encode both \\
%   * simple operational model of curried functions \\
%   * computations sequenced: effects \\
% \item useful for our purposes
%   * no need to extend lang with multi args / mut rec \\
%   * no cbn/cbv split: removes primary motivation for CPS?
% \item why untyped, open terms?
%   * exotic use of stack interesting for, e.g. varargs \\
%   * SOS on open terms used in interesting operational equivalences 
% \item examples
%   * varargs function
%   * eval combinators
% \end{itemize}

% \steve{Need to put a bit more structure on this section.  For now I've
%   tried paragraph headers.} \steve{Also could add examples}

\subsection{CBPV Structural Operational Semantics}
Figure~\ref{fig:cbpv} shows the syntax and operational semantics for
our variant of Levy's call-by-push-value calculus~\cite{levy},
which serves as the source language of the machine
equivalence.

As a functional language, CBPV is somewhat lower level and more
structured than ordinary lambda calculus. It syntactically
distinguishes values $\ottnt{V}$, which include variables $\ottmv{x}$ and
suspended computations $\kw{thunk} \, \ottnt{M}$, from general computation terms
$\ottnt{M}$. Computations can $ \kw{force} $ a thunk, produce a value
$\kw{prd} \, \ottnt{V}$, or sequence one computation before another
with $\ottnt{M} \, \kw{to} \, \ottmv{x} \, \kw{in} \, \ottnt{N}$, which also binds the result of $\ottnt{M}$ to the
variable $\ottmv{x}$ for use in $\ottnt{N}$.  Computation forms also include
conditional expressions, and arithmetic calculations, where we use
$ \oplus $ to stand for an arbitrary binary arithmetic
operation.  (In what follows, we write $ \means{\oplus} $ to denote the
semantic counterpart to $ \oplus $.)

Lambda abstractions $ \lambda \ottmv{x} . \ottnt{N} $, unlike in ordinary lambda calculus,
are \textit{not} values; they are computations. CBPV application
$ \ottnt{V} \!\cdot\! \ottnt{M} $ is written in reverse order, with the argument $\ottnt{V}$,
which is syntactically constrained to be a value, to the \textit{left}
of the function $\ottnt{M}$, which should evaluate to an abstraction
$ \lambda \ottmv{x} . \ottnt{N} $. Applications reduce via the usual $\beta$ rule
thereafter. An equivalent reading of these rules is to think of
$ \ottnt{V} \!\cdot\! \ottnt{M} $ as pushing $\ottnt{V}$ onto a stack while $\ottnt{M}$ continues to
compute until the $\beta$ step pops $V$ off the stack for use in the
function body. (This point of view is where call-by-push-value gets
its name.)

 The value--computation distinction is a key feature of the
CBPV design: its evaluation order is completely determined thanks to
syntactic restrictions that ensure there is never a choice between a
substitution step and a congruence rule.

\paragraph{Comparison to Levy's CBPV}

Our variant of CBPV includes a $ \kw{letrec} $ form that binds a bundle
of mutually recursive computation definitions for use in some scope.
Mutually recursive bindings are essential to practical programming
languages, but are often omitted in formal treatments of abstract
machines due to the complexity of dealing with mutually recursive
closures. We deviate from the usual treatment $ \kw{letrec} $, which
treats unrolling it as a computation step and instead define an
auxiliary relation $\leadsto$ that ``virtually'' unrolls $ \kw{letrec} $
in the course of reducing another redex. This treatment allows the SOS
to precisely match the steps of a CFG machine, where the scope of
labels is not explicitly represented, and stepping into scope does not
correspond to a step of execution.

%% CBPV
\begin{figure}[t]
  \[
  \begin{array}[t]{r@{ }l@{\ }r@{\  }ll@{}l}
    \mathrm{Values} \ \ni\ & \ottnt{V} & \bnf\ & \multicolumn{2}{l}{\ottmv{x}
                                             \sep \ottmv{n} \sep \kw{thunk} \, \ottnt{M}} \\

    \mathrm{Terms} \ \ni\ &  \ottnt{M},\ottnt{N} & \bnf\  & 
           \kw{force} \, \ottnt{V} & \sep \kw{letrec} \, \ottmv{x_{{\mathrm{1}}}}  \ottsym{=}  \ottnt{M_{{\mathrm{1}}}}  \ottsym{,} \, .. \, \ottsym{,}  \ottmv{x_{\ottmv{n}}}  \ottsym{=}  \ottnt{M_{\ottmv{n}}} \, \kw{in} \, \ottnt{N} \\
    && |\ & \kw{prd} \, \ottnt{V} & \sep \ottnt{M} \, \kw{to} \, \ottmv{x} \, \kw{in} \, \ottnt{N} \\
    && |\ &  \ottnt{V} \!\cdot\! \ottnt{M}     &  \sep  \lambda \ottmv{x} . \ottnt{M}  \\
    && |\ & \ottnt{V_{{\mathrm{1}}}}  \oplus  \ottnt{V_{{\mathrm{2}}}} & \sep  \kw{if0} \;  \ottnt{V} \;  \ottnt{M_{{\mathrm{1}}}} \;  \ottnt{M_{{\mathrm{2}}}}  \\
  \end{array}
  \]

  \[
  %% \infer{ \kw{prd} \, \ottnt{V} \leadsto \kw{prd} \, \ottnt{V} }{}\quad
  %% \infer{  \lambda \ottmv{x} . \ottnt{M}  \leadsto  \lambda \ottmv{x} . \ottnt{M}  }{}\quad
  \infer{ \ottnt{M} \leadsto \ottnt{M} }{}\quad
  \infer{ \kw{letrec} \, \ottmv{x_{{\mathrm{1}}}}  \ottsym{=}  \ottnt{M_{{\mathrm{1}}}}  \ottsym{,} \, .. \, \ottsym{,}  \ottmv{x_{\ottmv{n}}}  \ottsym{=}  \ottnt{M_{\ottmv{n}}} \, \kw{in} \, \ottnt{N} \leadsto \ottnt{N'} }
        { \{\kw{thunk} \, \ottsym{(}  \kw{letrec} \, \ottmv{x_{{\mathrm{1}}}}  \ottsym{=}  \ottnt{M_{{\mathrm{1}}}}  \ottsym{,} \, .. \, \ottsym{,}  \ottmv{x_{\ottmv{n}}}  \ottsym{=}  \ottnt{M_{\ottmv{n}}} \, \kw{in} \, \ottnt{M_{\ottmv{j}}}  \ottsym{)}/x_j\}_{j=1..n}N \leadsto N' }
  \]

  \begin{gather*}
  \infer{ \kw{force} \, \ottsym{(}  \kw{thunk} \, \ottnt{M}  \ottsym{)} \to \ottnt{M} }{}\quad 
  \infer{ \ottnt{M} \, \kw{to} \, \ottmv{x} \, \kw{in} \, \ottnt{N} \to \ottsym{\{}  \ottnt{V}  \ottsym{/}  \ottmv{x}  \ottsym{\}} \, \ottnt{N} }{\ottnt{M} \leadsto \kw{prd} \, \ottnt{V}}\quad
  \infer{  \ottnt{V} \!\cdot\! \ottnt{M}  \to \ottsym{\{}  \ottnt{V}  \ottsym{/}  \ottmv{x}  \ottsym{\}} \, \ottnt{N} }{\ottnt{M} \leadsto  \lambda \ottmv{x} . \ottnt{N} }\\
  \infer{ \ottnt{M} \, \kw{to} \, \ottmv{x} \, \kw{in} \, \ottnt{N} \to \ottnt{M'} \, \kw{to} \, \ottmv{x} \, \kw{in} \, \ottnt{N} }{\ottnt{M} \to \ottnt{M'}}\quad
  \infer{  \ottnt{V} \!\cdot\! \ottnt{M}  \to  \ottnt{V} \!\cdot\! \ottnt{M'}  }{\ottnt{M} \to \ottnt{M'}}\\
  \infer{ \ottnt{N} \to \ottnt{M} }
        { \ottnt{N} \leadsto \ottnt{N'} & \ottnt{N'} \to \ottnt{M}} \\
  \infer{ \ottnt{M} \, \kw{to} \, \ottmv{x} \, \kw{in} \, \ottnt{N} \to \ottsym{\{}   n_1   \means{\oplus}   n_2   \ottsym{/}  \ottmv{x}  \ottsym{\}} \, \ottnt{N}}{\ottnt{M} \leadsto \ottsym{(}   n_1   \oplus   n_2   \ottsym{)}}\qquad
  \infer{  \kw{if0} \;  \ottsym{0} \;  \ottnt{M_{{\mathrm{1}}}} \;  \ottnt{M_{{\mathrm{2}}}}  \to \ottnt{M_{{\mathrm{1}}}} }{} \qquad
  \infer{  \kw{if0} \;   n  \;  \ottnt{M_{{\mathrm{1}}}} \;  \ottnt{M_{{\mathrm{2}}}}  \to \ottnt{M_{{\mathrm{2}}}} }{}\ ( n \not = 0)
  %% \infer{ \kw{letrec} \, \ottmv{x_{\ottmv{i}}}  \ottsym{=}  \ottnt{M_{\ottmv{i}}} \, \kw{in} \, \ottnt{N} \to \ottnt{N'} }
  %%       { \ottsym{\{}  \kw{thunk} \, \ottsym{(}  \kw{letrec} \, \ottmv{x_{\ottmv{i}}}  \ottsym{=}  \ottnt{M_{\ottmv{i}}} \, \kw{in} \, \ottnt{M_{\ottmv{j}}}  \ottsym{)}  \ottsym{/}  \ottmv{x_{\ottmv{j}}}  \ottsym{\}} \, \ottnt{N} \to \ottnt{N'} }\\
  \end{gather*}

  \caption{Syntax and operational semantics for the CBPV
    language.}
  \label{fig:cbpv}
\end{figure}

This presentation of CBPV, unlike Levy's, is
\textit{untyped}. As a consequence, the operational semantics
presented in Figure~\ref{fig:cbpv} can get stuck trying to evaluate
ill-formed terms.  However, using the type system to rule out stuck
states is completely compatible with the correspondence with abstract
machines---a CBPV term is stuck exactly
when the abstract machine is, so type safety at the CBPV
level corresponds directly to type safety at the SSA level too.  (This
is an example of how results about CBPV can be translated to the
CFG level.) Moreover, using an untyped language makes it easier
to explore the semantics of features such as vararg functions: unlike
in the simply typed variant, it is possible to define a computation
that pops a dynamically computed number of arguments.

\subsection{CEK Machine}
\label{sec:cek}

\begin{figure*}

\[
\begin{array}[t]{r@{\ }l@{\ \ }c@{\ }lll}
  \mathrm{State} \ \ni\ & \sigma & =\  & \mathrm{Term} \times \mathrm{Env}
\times \mathrm{Kont} \\
  %% \mathrm{Term} \ \ni\ & \ottnt{c} & \bnf\ & \ottnt{M} \\
  %% \mathrm{Env} \ \ni\ & \ottnt{e} & =\ & \mathrm{Var}
  %%   \longrightarrow_{\mathrm{fin}} \mathrm{Val} \\
  \mathrm{Env} \ \ni\ & \ottnt{e} & \bnf\ & \cdot
    \sep \ottsym{[}  \ottmv{x}  \mathbin{\mapsto}  \ottnt{v}  \ottsym{]}  \ottnt{e}
    \sep \ottsym{[}  \ottmv{x_{{\mathrm{1}}}}  \mathbin{\mapsto}  \ottnt{M_{{\mathrm{1}}}}  \ottsym{,} \, .. \, \ottsym{,}  \ottmv{x_{\ottmv{n}}}  \mathbin{\mapsto}  \ottnt{M_{\ottmv{n}}}  \ottsym{]}  \ottnt{e}\\
  \mathrm{Kont} \ \ni\ & k & \bnf\ & \cdot \sep  \ottnt{v}  \mathbin{\cdot} \, \underline{\ \ } :: \ottnt{k} 
    \sep  \ottsym{[} \, \underline{\ \ } \, \kw{to} \, \ottmv{x} \, \kw{in} \, \ottnt{M}  \ottsym{,}  \ottnt{e}  \ottsym{]} :: \ottnt{k}  \\
  %% TODO: errors?
  \mathrm{DVal} \ \ni\ & v & \bnf\ & \ottmv{x} \sep \ottsym{[}  \ottnt{M}  \ottsym{,}  \ottnt{e}  \ottsym{]}
\end{array}
\]

\[
\begin{array}{rcl}
  \gamma & : & \mathrm{Val} \rightarrow \mathrm{Env} \rightarrow \mathrm{DVal} \\
   \gamma\;  \ottmv{x} \;  \ottnt{e}  & = & \ottnt{e}  \ottsym{(}  \ottmv{x}  \ottsym{)} \\
  %%  \gamma\;   n  \;  \ottnt{e}  & = & \ottmv{n} \\
   \gamma\;  \ottsym{(}  \kw{thunk} \, \ottnt{M}  \ottsym{)} \;  \ottnt{e}  & = & \ottsym{[}  \ottnt{M}  \ottsym{,}  \ottnt{e}  \ottsym{]} \\
\end{array}
\]

%% \tiny
\begin{gather*}
\infer
    {  \langle \kw{force} \, \ottnt{V}  ,\;  \ottnt{e}  ,\;  \ottnt{k} \rangle 
       \longrightarrow 
       \langle \ottnt{M}  ,\;  \ottnt{e'}  ,\;  \ottnt{k} \rangle  }
    {  \gamma\;  \ottnt{V} \;  \ottnt{e}  = \ottsym{[}  \ottnt{M}  \ottsym{,}  \ottnt{e'}  \ottsym{]} }
\quad
\infer
    {  \langle  \lambda \ottmv{x} . \ottnt{M}   ,\;  \ottnt{e}  ,\;   \ottsym{(}  \ottnt{v}  \mathbin{\cdot} \, \underline{\ \ }  \ottsym{)} :: \ottnt{k}  \rangle 
       \longrightarrow 
       \langle \ottnt{M}  ,\;  \ottsym{[}  \ottmv{x}  \mathbin{\mapsto}  \ottnt{v}  \ottsym{]}  \ottnt{e}  ,\;  \ottnt{k} \rangle  }
    { }
\\
\infer 
    {  \langle \kw{prd} \, \ottnt{V}  ,\;  \ottnt{e}  ,\;   \ottsym{[} \, \underline{\ \ } \, \kw{to} \, \ottmv{x} \, \kw{in} \, \ottnt{M}  \ottsym{,}  \ottnt{e'}  \ottsym{]} :: \ottnt{k}  \rangle  
       \longrightarrow 
       \langle \ottnt{M}  ,\;  \ottsym{[}  \ottmv{x}  \mathbin{\mapsto}  \ottnt{v}  \ottsym{]}  \ottnt{e'}  ,\;  \ottnt{k} \rangle  }
    {  \gamma\;  \ottnt{V} \;  \ottnt{e}  = \ottnt{v} }
\\
\infer
    {  \langle  \ottnt{V} \!\cdot\! \ottnt{M}   ,\;  \ottnt{e}  ,\;  \ottnt{k} \rangle 
       \longrightarrow  \sigma }
    {  \gamma\;  \ottnt{V} \;  \ottnt{e}  = \ottnt{v} &
       \langle \ottnt{M}  ,\;  \ottnt{e}  ,\;   \ottsym{(}  \ottnt{v}  \mathbin{\cdot} \, \underline{\ \ }  \ottsym{)} :: \ottnt{k}  \rangle 
       \longrightarrow  \sigma }
\quad
\infer
    {  \langle \ottnt{M} \, \kw{to} \, \ottmv{x} \, \kw{in} \, \ottnt{N}  ,\;  \ottnt{e}  ,\;  \ottnt{k} \rangle 
       \longrightarrow  \sigma }
    {  \langle \ottnt{M}  ,\;  \ottnt{e}  ,\;   \ottsym{[} \, \underline{\ \ } \, \kw{to} \, \ottmv{x} \, \kw{in} \, \ottnt{N}  \ottsym{,}  \ottnt{e}  \ottsym{]} :: \ottnt{k}  \rangle 
       \longrightarrow  \sigma }
\\
\infer
    {  \langle \kw{letrec} \, \ottmv{x_{\ottmv{i}}}  \ottsym{=}  \ottnt{M_{\ottmv{i}}} \, \kw{in} \, \ottnt{N}  ,\;  \ottnt{e}  ,\;  \ottnt{k} \rangle 
       \longrightarrow  \sigma }
    {  \langle \ottnt{N}  ,\;  \ottsym{[}  \ottmv{x_{\ottmv{j}}}  \mathbin{\mapsto}  \kw{letrec} \, \ottmv{x_{\ottmv{i}}}  \ottsym{=}  \ottnt{M_{\ottmv{i}}} \, \kw{in} \, \ottnt{M_{\ottmv{j}}}  \ottsym{]}  \ottnt{e}  ,\;  \ottnt{k} \rangle 
       \longrightarrow  \sigma }
\end{gather*}

  \caption{CEK Machine}
  \label{fig:cek}
\end{figure*}

As the first step toward control flow graphs, we define a CEK machine~\cite{cek} for CBPV terms, the syntax of which is
given in Figure ~\ref{fig:cek}. States of a CEK machine are triples $ \langle \ottnt{c}  ,\;  \ottnt{e}  ,\;  \ottnt{k} \rangle $ where $\ottnt{c}$ is a CBPV term (here we use the metavariable
$\ottnt{c}$ rather than $\ottnt{M}$ to emphasize that this term lives at the CEK
level), $\ottnt{e}$ is an environment containing variable bindings, and $\ottnt{k}$
is a stack of continuation frames. CEK values $v$ include variables,
number constants, and, unlike in CBPV, \textit{closures} consisting of the
code of a thunk and a captured environment, written $\ottsym{[}  \ottnt{M}  \ottsym{,}  \ottnt{e}  \ottsym{]}$. In CBPV,
there are two search rules in the SOS, and so we have two kinds of frames:
application and sequence, each written with a ``hole'' $\underline{\ \ }$ in the place
where the sub-expression's produced value will be placed.

Our CEK machine differs from the usual definition in two significant
ways: environments have more structure than the usual map of variables
to values, and we include recursive rules for stepping under
applications, sequencing, and recursive bindings. When execution
encounters a $ \kw{letrec} $, the environment is extended with the
associated terms without creating closures. The intuition is that,
since within the scope of the variables being bound the environment
can only be extended, it suffices to mark the environment instead of
making explicit copies in a closure. The required closures can be
created on-demand when the associated variable is referenced by the
lookup operation $\ottnt{e}  \ottsym{(}  \ottmv{x}  \ottsym{)}$. This requires the environment to have
additional structure: it is a sequence of bindings rather than a
simple mapping. The major benefit is that forcing a variable bound by
$ \kw{letrec} $ never allocates an intermediate closure.

The definition of the CEK machine also typically executes the
search rules of the SOS as a series of smaller steps. Executing under
an application, for example, pushes an argument frame onto the stack
in a distinct step of the machine. The reasoning is that recursive
rules are hard to implement in hardware or an efficient subset of a
low-level language like C. We take a different approach and leave the
recursive rules, relying on a subsequent partial evaluation step to
produce an easily-implementable machine. A series of applications
will, for example, correspond to a multi-argument call instruction.

Figure~\ref{fig:cek} gives the operational sematics of the CEK machine. It
relies on a function $\gamma$, shown in Figure~\ref{fig:cek} for looking
up values in the environment, which constructs closures as necessary. We
write $\ottnt{e}  \ottsym{(}  \ottmv{x}  \ottsym{)}$ for the value associated with variable $\ottmv{x}$ in
environment $\ottnt{e}$, and $\ottsym{[}  \ottmv{x}  \mathbin{\mapsto}  \ottnt{v}  \ottsym{]}  \ottnt{e}$ for the environment $\ottnt{e}$ updated
to map $\ottmv{x}$ to $\ottnt{v}$. Each search construct pushes a corresponding
frame onto the stack, while the corresponding reduction step pops the
frame off the stack and uses it as a context for reduction.

\paragraph{Unloading CEK Machine States and Correctness.}
The correctness of a CEK machine with respect to the SOS of the source
language is usually stated in terms of the associated evaluation
functions for closed terms terminating at non-function constants. The
proof involves a simulation argument using a mapping between machine
states and terms. Since we have ensured that the steps of our SOS and
CEK machine line up exactly, we can prove a stronger property stated
directly in terms of the functions $ \kw{CEK.load} $ and $ \kw{CEK.unload} $,
\footnote{We elide the
  formal definition of \kw{unload} for the CEK machine which is
  more-or-less standard, and instead present the similar but more
  complicated function for the PEAK machine in
  Figure~\ref{fig:pekunlfuns}.}
mapping terms to CEK states and vice versa:

\begin{lemma} (CEK Correctness) For all CEK states $\sigma$, 
\[
 \kw{CEK.unload} (\kw{CEK.load} \, \ottnt{M}) = M
\]
\[
 \kw{CBPV.step} (\kw{CEK.unload} \, \sigma) = 
 \kw{CEK.unload} (\kw{CEK.step} \, \sigma)
\]
\end{lemma}

Loading a term into a machine state maps a CBPV term $M$ to the CEK
state $ \langle \ottnt{c}  ,\;  \cdot  ,\;  \cdot \rangle $, while unloading carries out the delayed
substitutions in the environment and reconstructs the context from the
continuation frames. Since our environments contain both values and
$ \kw{letrec} $ bindings, we must first recover the associated closures
by ``flattening'' the environment, making the explicit copies for
closures that were elided during execution. Closures are then unloaded
recursively into terms in the usual way.

All of our machine correctness results follow the above form. The
first equation guarantees that states of the ``more abstract'' machine
correspond to sets of states of the ``more concrete'' machine. The
second equation shows that, in fact, the transition system defined by
the more concrete machine, considered up to the equivalence relation
induced by the first equation, is the same one as the more abstract
machine.

%% Local Variables:
%% fill-column: 70
%% eval: (auto-fill-mode)
%% eval: (flyspell-mode)
%% eval: (setq-local auto-hscroll-mode nil)
%% eval: (setq-local sentence-end-double-space nil)
%% End:

\section {Proving Machine Equivalence}
\label{sec:machines}
\label{sec:proof}

\begin{figure}[t]
\[
\small
%% \resizebox{\textwidth}{!}{%
%% \tiny
\xymatrix@C=14pt{
\mbox{CBPV} & & M \ar[rrrrrr]^{\kw{CBPV.step}} && && && M' \\
\mbox{CEK}  
& 
& \sigma \ar@{->>}[u] \ar[rrrrrr]^{\kw{CEK.step}}
&& && && \sigma' \ar@{->>}[u]_{\kw{CEK.unload}} \\
\mbox{PEAK}  
&   P \vdash 
& \rho \ar@{->>}[u] \ar[rrrrrr]^{\kw{PEAK.step}(P)}
&& && && \rho' \ar@{->>}[u]_{\kw{PEAK.unload}(P)} \\ 
\mbox{PEK}  
&   P \vdash \ar[d]_{\kw{compile}}
& s \ar@{->>}[u] \ar[rrrrrr]^{\kw{PEK.step}(P)}
&& && && s' \ar@{->>}[u]_{\kw{PEK.unload}(P)} \\ 
\mbox{CFG}   
&   G \vdash 
& s \ar@{->>}[u] \ar[rrrrrr]^{\kw{CFG.step}(G)}
&& && && {s'} \ar@{->>}[u]_{\kw{CFG.unload}(P)} \\ 
}
%% }
\]
  \caption{Structure of the machine correspondences.}
  \label{fig:proofs}
\end{figure}
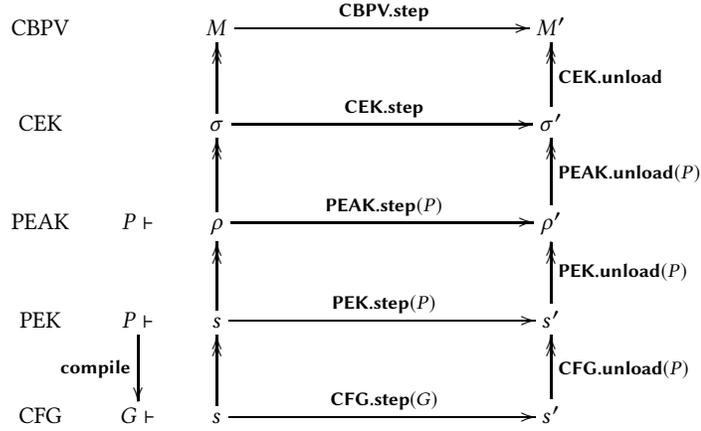

%% \resizebox{\textwidth}{!}{%
%% \tiny
%% \xymatrix@C=14pt{
%% \mbox{CBPV} & & M \ar[rrrrrr]^{\kw{CBPV.step}} && && && M' \\
%% \mbox{CEK}  
%% & 
%% &  \langle \ottnt{c}  ,\;  \ottnt{e}  ,\;  \ottnt{k} \rangle  \ar@{->>}[u] \ar[rrrrrr]^{\kw{CEK.step}}
%% && && && \sigma \ar@{->>}[u]_{\kw{CEK.unload}} \\
%% \mbox{PEAK}  
%% &   P \vdash 
%% &  \langle \ottnt{p}  ,\;  e  ,\;  a  ,\;  k \rangle  \ar@{->>}[u] \ar[rrrrrr]^{\kw{PEAK.step}(P)}
%% && && && \rho \ar@{->>}[u]_{\kw{PEAK.unload}(P)} \\ 
%% \mbox{PEK}  
%% %% & P \vdash
%% &   P \vdash \ar[d]^{\kw{compile}}
%%  %% \ar@/_1.5pc/[dd]|<>(0.8){\texttt{\textcolor{red}{<<no parses (char 8): compile1*** >>}}} 
%% &  \langle \ottnt{p}  ,\;  e  ,\;  k \rangle  \ar@{->>}[u] \ar[rrrrrr]^{\kw{PEK.step}(P)}
%% %% && *++{\rho_1} \ar@{_{(}->}[u] \ar[rr] 
%% %% && *++{\rho_2} \ar@{_{(}->}[u] \ar@{-->}[rr]
%% && && && s \ar@{->>}[u]_{\kw{PEK.unload}(P)} \\ 
%% \mbox{CFG}   
%% &   G \vdash 
%% &  \langle \ottnt{p}  ,\;  e  ,\;  k \rangle  \ar@{->>}[u] \ar[rrrrrr]^{\kw{CFG.step}(G)}
%% %% && *++{s_1} \ar@{_{(}->}[u] \ar[rr] 
%% %% && *++{s_2} \ar@{_{(}->}[u] \ar@{-->}[rr] 
%% && && && *++{s} \ar@{->>}[u]_{\kw{CFG.unload}(P)} \\ 
%% %% \mbox{CFG} 
%% %% & \mathit{cfg_1} \vdash 
%% %% &  \langle \ottnt{p}  ,\;  e  ,\;  k \rangle  \ar[rrrrrr]^{ \kw{CFG.step} }  \ar@{_{(}->}[u]|{\kw{unload}}
%% %% && && &&  \langle \ottnt{p'}  ,\;  e'  ,\;  k' \rangle  \ar@{_{(}->}[u]|{\kw{unload}} \\ 
%% }}

Figure~\ref{fig:proofs} shows the structure of the machine
correspondences proved in this paper. We translate call-by-push-value
terms through a series of abstract machines, eventually reaching an
SSA CFG language. At each level, we prove a correctness lemma that relates the
behavior of states of the current machine to the behavior of
corresponding states in the previous machine. For most of these
levels, the correctness statement is lockstep: each step in the
current machine matches exactly one step in the previous machine. At
the top and bottom levels, multiple steps in the intermediate machines
may correspond to a single step in a CBPV term or a CFG machine. The
relation is something more precise than stuttering simulation,
however; by examining the state involved, we can compute precisely the
number of intermediate-machine steps that will be required to make up
a single CBPV/CFG step.

More specifically, we define a function \kw{unload} that translates
CFG states into CBPV terms, defined as the composition of simpler \kw{unload}
functions between each of the layers. We then define a function
\kw{depth} such that each CFG machine step of a state $s$ corresponds
to exactly $\kw{depth}(s)$ steps of the last intermediate machine, 
and $\kw{depth}(s)$ steps of the first intermediate machine
(the CEK machine of the previous section) correspond in turn to one step of $\kw{unload}(s)$.
(In fact, the function $\kw{depth}$ is precisely the redex depth of the term being processed,
which we have already seen gives us the number of steps of the CEK machine that make up one step of the SOS.) Since the
steps of each intermediate machine are in lockstep correspondence, we
can conclude that each step of the CFG machine corresponds exactly to
a step of the corresponding CBPV term.

This conclusion still leaves open the possibility that not every CBPV term
corresponds to a CFG state. Fortunately, it is easy to construct a state
corresponding to any given CBPV term. We can define a function \kw{load}
at each level that translates a CBPV term into a machine state. For instance,
we can load a CBPV term $\ottnt{M}$ into the CEK machine state $ \langle \ottnt{M}  ,\;  \cdot  ,\;  \cdot \rangle $.
At every level, it is the case that $\kw{unload}(\kw{load}(M)) = M$,
guaranteeing that every CBPV term has at least one corresponding CFG state.

We also define a notion of well-formedness at each level. The machines may
contain states that do not correspond to CBPV terms; the well-formedness
property rules out these states, by translating the scoping constraints of
terms into constraints on machine states. We show that initial states
produced by loading are well-formed, well-formedness is preserved by
steps of the machine, and well-formed states unload into terms. Together,
these properties suffice to guarantee that starting from a CBPV term,
we only encounter machine states that correspond to CBPV terms.

\subsection{Machine equivalence}

\todo{Dmitri should read this section.}
\steve{Per our phone call:  We should foreshadow a deeper discussion
  of applications/ equivalences that will appear in a later section.}
The net result is stated in Theorem~\ref{compile-correct}.

\begin{theorem}\label{compile-correct} (CBPV/CFG Machine Equivalence) For
  all CBPV  terms $\ottnt{P}$ and CFG states $s$ that are well-formed with respect
  to $\ottnt{P}$, %(under the improved compilation strategy) 
  \[
  \kw{CBPV.step} \, \ottsym{(}  \kw{unload} \, \ottnt{P} \, s  \ottsym{)} = \kw{unload} \, \ottnt{P} \, \ottsym{(}  \kw{CFG.step} \, \ottsym{(}  \kw{compile} \, \ottnt{P}  \ottsym{)} \, s  \ottsym{)}
  \]
\end{theorem}
\begin{proof}By combining the per-machine correctness lemmas, as shown in Figure~\ref{fig:proofs}.\end{proof}

This theorem gives strong reasoning properties about the connection
between CBPV terms and their corresponding machine semantics.  As one
illustration of its use, consider an equivalence on closed CBPV terms
$\ottnt{M_{{\mathrm{1}}}}$ and $\ottnt{M_{{\mathrm{2}}}}$ such that
\[
\ottnt{M_{{\mathrm{1}}}} \equiv \ottnt{M_{{\mathrm{2}}}} \qquad \mathit{iff} \qquad  (\ottnt{M_{{\mathrm{1}}}} \to^*
\kw{prd} \, \ottsym{0} \AND \ottnt{M_{{\mathrm{2}}}} \to^* \kw{prd} \, \ottsym{0})
\]
We can define what it means for two CFG programs to be equivalent by 
\[
\mathit{cfg}_{{\mathrm{1}}} \equiv' \mathit{cfg}_{{\mathrm{2}}} \quad \mathit{iff} 
\left\{
  \begin{array}{l}
    \mathit{cfg}_{{\mathrm{1}}} \vdash s_{{\mathrm{0}}} \to^* s_{{\mathrm{1}}} \not \to \\
    \mathit{cfg}_{{\mathrm{2}}} \vdash s_{{\mathrm{0}}} \to^* s_{{\mathrm{2}}} \not \to \\
    \mbox{and}\ \kw{ret}(s_{{\mathrm{1}}}) = \kw{0} = \kw{ret}(s_{{\mathrm{2}}})
  \end{array}
\right.
\]
\noindent Here, the function $\kw{ret}(s)$ examines a CFG machine state to
determine whether it has terminated at a $ \kwl{RET} $ instruction whose
operand evaluates to \kw{0} in the state $s$, where $s_{{\mathrm{0}}}$ denotes
the initial CFG state.  An immediate corollary of
Theorem~\ref{compile-correct} is:
\[
\ottnt{M_{{\mathrm{1}}}} \equiv \ottnt{M_{{\mathrm{2}}}} \IMPLIES (\kw{compile} \, \ottnt{M_{{\mathrm{1}}}} \equiv' \kw{compile} \, \ottnt{M_{{\mathrm{2}}}})
\]
\noindent Note that if $\ottnt{M_{{\mathrm{1}}}} \to \ottnt{M'_{{\mathrm{1}}}}$ then $\ottnt{M_{{\mathrm{1}}}} \equiv
\ottnt{M'_{{\mathrm{1}}}}$, so we immediately conclude that $\kw{compile} \, \ottnt{M_{{\mathrm{1}}}} \equiv'
\kw{compile} \, \ottnt{M'_{{\mathrm{1}}}}$.  Many other optimizations also fall into this style
of reasoning.

%% Importantly, the machine equivalence theorem holds even for
%% \textit{open} CBPV terms.   We conjecture that this means that
%% reasoning similar to the above can be generalized to contextual
%% equivalences as well.

% The CBPV language is naturally equipped with a structural operational
% semantics defined in terms of substitution along the lines of the
% usual lambda calculus.  This gives rise to a rich equational theory. 

% CFG machine has a precise semantics in terms of an environment and
% stack machine.  

% Key result: show that the CBPV language and the CFG machine define
% \textit{identical} transition systems.

%% Local Variables:
%% fill-column: 70
%% eval: (auto-fill-mode)
%% eval: (flyspell-mode)
%% eval: (setq-local auto-hscroll-mode nil)
%% eval: (setq-local sentence-end-double-space nil)
%% End:

\section{The PEAK Machine}
\label{sec:peak}
%% Rationale
The CEK machine gives a straightforward computation model for CBPV,
but it falls short of the execution model
of our desired CFG formalism. We call the next step the PEAK machine.

The principle underlying the PEAK machine's execution model is that
the CEK machine does unnecessary work by saving an environment
immediately when an expression is evaluated. By delaying closure
allocation until the environment is actually changed, we avoid some
paired save/restore operations. Furthermore, using an alternative
representation of environments allows the machine to continue to add
bindings to the environment as as long as the result of substitution
does not change. As a result, the PEAK machine can express
executions that contain control flow join
points, like the CFG program in Figure~\ref{fig:ssacfg}, without the need to save and restore an environment.

%% PEAK

\begin{figure}
%  
%% TODO: integrate into section; align some seps; all sorts of syntax changes
\[
\begin{array}[t]{r@{\ }l@{\ \ }c@{\ }lll}
  \mathrm{State} \ \ni\ & \rho & =\  & \mathrm{Path} \times \mathrm{Env}
    \times \mathrm{Args} \times \mathrm{Kont} \\
  \mathrm{Path} \ \ni\ & \ottnt{p}, \ottnt{q} & \bnf\ & \cdot \sep  \ottmv{n}  \CONS  \ottnt{p}  \\
  \mathrm{Env} \ \ni\ & e & =\ & \mathrm{Path}
    \longrightarrow_{\mathrm{fin}} \mathrm{Val} \\
  \mathrm{Args} \ \ni\ & a & \bnf\ & \cdot \sep   \kwl{ARG}\; \ottnt{p}   ::  a 
    \sep  \kwl{SEQ} \, \ottnt{p}  ::  a  \\
  \mathrm{Kont} \ \ni\ & k & \bnf\ & \cdot \sep  v  \mathbin{\cdot} \, \underline{\ \ } :: k 
    \sep  \ottsym{[}  \ottnt{p}  \ottsym{,}  e  \ottsym{,}  a  \ottsym{]} :: k  \\
  %% TODO: errors?
  \mathrm{Val} \ \ni\ & v & \bnf\ & \ottmv{x} \sep \ottmv{n} \sep \ottsym{[}  \ottnt{p}  \ottsym{,}  e  \ottsym{]}
\end{array}
\]

  \[
  \begin{array}{rcll}
    \gamma & : & \mathrm{Path} \rightarrow \mathrm{Env} \rightarrow \mathrm{DVal} \\
    \gamma \, \ottnt{p} \, e & = & \mathrm{lookup}(\ottnt{P}, \ottnt{p}, \ottnt{e}, \ottmv{x})
                          & \mbox{when $ \ottnt{P}  [  \ottnt{p}  ]  = \ottmv{x}$}

    \\
    %% \gamma \, \ottnt{p} \, e & = & e  \ottsym{(}  \ottnt{p'}  \ottsym{)}
    %%                       & \mbox{when $ \ottnt{P}  [  \ottnt{p}  ]  = \ottmv{x}$,}\\
    %%                 & & &  \mbox{\quad $\ottmv{x}$ bound by $ \ottnt{P}  [  \ottnt{p'}  ]  =  \lambda \ottmv{x} . \underline{\ \ }  $ or} \\
    %%                 & & & \mbox{\quad by $ \ottnt{P}  [  \ottnt{p'}  ]  = \underline{\ \ } \, \kw{to} \, \ottmv{x} \, \kw{in} \, \underline{\ \ } $} 

    %% \\
    %% \gamma \, \ottnt{p} \, e & = & \ottsym{[}   \ottsym{0}  \CONS  \ottnt{p'}   \ottsym{,}  e  \ottsym{]}
    %%                     & \mbox{when $ \ottnt{P}  [  \ottnt{p}  ]  = \ottmv{x}$, $\ottmv{x}$ bound} \\
    %%                     &  & & \mbox{\quad by $ \ottnt{P}  [  \ottnt{p'}  ]  = \texttt{\textcolor{red}{<<no parses (char 1): r***ec x. \_ >>}}$} 

    %% \\                            
    \gamma \, \ottnt{p} \, e & = & \ottsym{[}   \ottsym{0}  \CONS  \ottnt{p}   \ottsym{,}  e  \ottsym{]} 
                        & \mbox{when $ \ottnt{P}  [  \ottnt{p}  ]  = \kw{thunk} \, \underline{\ \ }$}

    \\
     \gamma \, \ottnt{p} \, e & = & \ottmv{n} 
                         & \mbox{when $ \ottnt{P}  [  \ottnt{p}  ]  = \ottmv{n}$} 
\end{array}
\]
\[
  \begin{array}{rcll}
    \delta & : & \mathrm{Env} \rightarrow \mathrm{Args} \rightarrow \mathrm{Kont} \\
     \delta\;  e  \;  \cdot  & = & . \\
     \delta\;  e  \;  \ottsym{(}    \kwl{ARG}\; \ottnt{p}   ::  a   \ottsym{)}  & = &  \ottsym{(}  \gamma \, \ottsym{(}   \ottsym{0}  \CONS  \ottnt{p}   \ottsym{)} \, e  \mathbin{\cdot} \, \underline{\ \ }  \ottsym{)} ::  \delta\;  e  \;  a   \\
     \delta\;  e  \;  \ottsym{(}   \kwl{SEQ} \, \ottnt{p}  ::  a   \ottsym{)}  & = &  \ottsym{[}  \ottnt{p}  \ottsym{,}  e  \ottsym{,}  a  \ottsym{]} :: \cdot  \\                                     
    %%  \delta\;  e  \;  \ottsym{(}   \kwl{PRJ} \, \ottmv{i}  ::  a   \ottsym{)}  & = &   \kw{prj}_{ \ottmv{i} }\; \underline{\ \ }  ::  \delta\;  e  \;  a   \\
\end{array}
  \]

  \caption{PEAK syntax and semantic functions.}
  \label{fig:peakfun}
\end{figure}

\begin{figure}

\begin{gather*}
\infer
   {  \langle \ottnt{p}  ,\;  e  ,\;  a  ,\;  k \rangle 
      \longrightarrow 
      \langle \ottnt{p'}  ,\;  e'  ,\;  \cdot  ,\;   \delta\;  e  \;  a   +\!\!\!+  k \rangle  }
   {  \ottnt{P}  [  \ottnt{p}  ]  = \kw{force} \,  \underline{\ \ } _{ \ottsym{0} }  &
     \gamma \, \ottsym{(}   \ottsym{0}  \CONS  \ottnt{p}   \ottsym{)} \, e = \ottsym{[}  \ottnt{p'}  \ottsym{,}  e'  \ottsym{]} }
\\
\infer
    {  \langle \ottnt{p}  ,\;  e  ,\;   \kwl{SEQ} \, \ottnt{p'}  ::  a   ,\;  k \rangle 
       \longrightarrow 
       \langle  \ottsym{1}  \CONS  \ottnt{p'}   ,\;  e  \ottsym{[}  \ottnt{p'}  \mathbin{\mapsto}  v  \ottsym{]}  ,\;  a  ,\;  k \rangle  }
    {  \ottnt{P}  [  \ottnt{p}  ]  = \kw{prd} \,  \underline{\ \ } _{ \ottsym{0} }  &
      \gamma \, \ottsym{(}   \ottsym{0}  \CONS  \ottnt{p}   \ottsym{)} \, e = v }
\\
\infer
    {  \langle \ottnt{p}  ,\;  e  ,\;  \cdot  ,\;   \ottsym{[}  \ottnt{p'}  \ottsym{,}  e'  \ottsym{,}  a'  \ottsym{]} :: k  \rangle 
       \longrightarrow 
       \langle  \ottsym{1}  \CONS  \ottnt{p'}   ,\;  e'  \ottsym{[}  \ottnt{p'}  \mathbin{\mapsto}  v  \ottsym{]}  ,\;  a'  ,\;  k \rangle  }
    {  \ottnt{P}  [  \ottnt{p}  ]  = \kw{prd} \,  \underline{\ \ } _{ \ottsym{0} }  &
      \gamma \, \ottsym{(}   \ottsym{0}  \CONS  \ottnt{p}   \ottsym{)} \, e = v }
\\
\infer
    {  \langle \ottnt{p}  ,\;  e  ,\;    \kwl{ARG}\; \ottnt{q}   ::  a   ,\;  k \rangle 
       \longrightarrow 
       \langle  \ottsym{0}  \CONS  \ottnt{p}   ,\;  e  \ottsym{[}  \ottnt{p}  \mathbin{\mapsto}  v  \ottsym{]}  ,\;  a  ,\;  k \rangle  }
    {  \ottnt{P}  [  \ottnt{p}  ]  =  \lambda \ottmv{x} .  \underline{\ \ } _{ \ottsym{0} }   &
      \gamma \, \ottsym{(}   \ottsym{0}  \CONS  \ottnt{q}   \ottsym{)} \, e = v }
\\
\infer
    {  \langle \ottnt{p}  ,\;  e  ,\;  \cdot  ,\;   v  \mathbin{\cdot} \, \underline{\ \ } :: k  \rangle 
       \longrightarrow 
       \langle  \ottsym{0}  \CONS  \ottnt{p}   ,\;  e  \ottsym{[}  \ottnt{p}  \mathbin{\mapsto}  v  \ottsym{]}  ,\;  \cdot  ,\;  k \rangle  }
    {  \ottnt{P}  [  \ottnt{p}  ]  =  \lambda \ottmv{x} .  \underline{\ \ } _{ \ottsym{0} }   }
\\
\infer
    {  \langle \ottnt{p}  ,\;  e  ,\;  a  ,\;  k \rangle 
       \longrightarrow  \rho }
    {  \ottnt{P}  [  \ottnt{p}  ]  =  \underline{\ \ } _{ \ottsym{0} }  \, \kw{to} \, \ottmv{x} \, \kw{in} \,  \underline{\ \ } _{ \ottsym{1} }  &
       \langle  \ottsym{0}  \CONS  \ottnt{p}   ,\;  e  ,\;   \kwl{SEQ} \, \ottnt{p}  ::  a   ,\;  k \rangle 
       \longrightarrow  \rho }
\\
\infer
    {  \langle \ottnt{p}  ,\;  e  ,\;  a  ,\;  k \rangle 
       \longrightarrow  \rho }
    {  \ottnt{P}  [  \ottnt{p}  ]  =   \underline{\ \ } _{ \ottsym{0} }  \!\cdot\!  \underline{\ \ } _{ \ottsym{1} }   &
       \langle  \ottsym{1}  \CONS  \ottnt{p}   ,\;  e  ,\;    \kwl{ARG}\; \ottnt{p}   ::  a   ,\;  k \rangle 
       \longrightarrow  \rho }
\\
\infer
    {  \langle \ottnt{p}  ,\;  e  ,\;  a  ,\;  k \rangle 
       \longrightarrow  \rho }
    {  \ottnt{P}  [  \ottnt{p}  ]  = \kw{letrec} \, \ottmv{x_{\ottmv{i}}}  \ottsym{=}   \underline{\ \ } _{ \ottmv{i} }  \, \kw{in} \,  \underline{\ \ } _{ \ottsym{0} }  &
       \langle  \ottsym{0}  \CONS  \ottnt{p}   ,\;  e  ,\;  a  ,\;  k \rangle 
       \longrightarrow  \rho }
\\
\infer
    {  \langle \ottnt{p}  ,\;  e  ,\;  a  ,\;  k \rangle 
       \longrightarrow   \langle  \ottsym{1}  \CONS  \ottnt{p}   ,\;  e  ,\;  a  ,\;  k \rangle  }
    {  \ottnt{P}  [  \ottnt{p}  ]  =  \kw{if0} \;   \underline{\ \ } _{ \ottsym{0} }  \;   \underline{\ \ } _{ \ottsym{1} }  \;   \underline{\ \ } _{ \ottsym{2} }   &
      \gamma \,  \ottsym{0}  \CONS  \ottnt{p}  \, e = \ottsym{0}}
\\
\infer
    {  \langle \ottnt{p}  ,\;  e  ,\;  a  ,\;  k \rangle 
       \longrightarrow   \langle  \ottsym{2}  \CONS  \ottnt{p}   ,\;  e  ,\;  a  ,\;  k \rangle  }
    {  \ottnt{P}  [  \ottnt{p}  ]  =  \kw{if0} \;   \underline{\ \ } _{ \ottsym{0} }  \;   \underline{\ \ } _{ \ottsym{1} }  \;   \underline{\ \ } _{ \ottsym{2} }   &
      \gamma \,  \ottsym{0}  \CONS  \ottnt{p}  \, e  \neq  \ottsym{0}}
\\
\infer
    {  \langle \ottnt{p}  ,\;  e  ,\;   \kwl{SEQ} \, \ottnt{p'}  ::  a   ,\;  k \rangle 
       \longrightarrow 
       \langle  \ottsym{1}  \CONS  \ottnt{p'}   ,\;  e  \ottsym{[}  \ottnt{p'}  \mathbin{\mapsto}  v_{{\mathrm{1}}}  \means{\oplus}  v_{{\mathrm{2}}}  \ottsym{]}  ,\;  a  ,\;  k \rangle  }
    {  \ottnt{P}  [  \ottnt{p}  ]  =  \underline{\ \ } _{ \ottsym{0} }   \oplus   \underline{\ \ } _{ \ottsym{1} }  &
       \gamma \, \ottsym{(}   \ottsym{0}  \CONS  \ottnt{p}   \ottsym{)} \, e = v_{{\mathrm{1}}} &
       \gamma \, \ottsym{(}   \ottsym{1}  \CONS  \ottnt{p}   \ottsym{)} \, e = v_{{\mathrm{2}}} }
\\
\infer
    {  \langle \ottnt{p}  ,\;  e  ,\;  \cdot  ,\;   \ottsym{[}  \ottnt{p'}  \ottsym{,}  e'  \ottsym{,}  a'  \ottsym{]} :: k  \rangle 
       \longrightarrow 
       \langle  \ottsym{1}  \CONS  \ottnt{p'}   ,\;  e'  \ottsym{[}  \ottnt{p'}  \mathbin{\mapsto}  v_{{\mathrm{1}}}  \means{\oplus}  v_{{\mathrm{2}}}  \ottsym{]}  ,\;  a'  ,\;  k \rangle  }
    {  \ottnt{P}  [  \ottnt{p}  ]  =  \underline{\ \ } _{ \ottsym{0} }   \oplus   \underline{\ \ } _{ \ottsym{1} }  &
       \gamma \, \ottsym{(}   \ottsym{0}  \CONS  \ottnt{p}   \ottsym{)} \, e = v_{{\mathrm{1}}} &
       \gamma \, \ottsym{(}   \ottsym{1}  \CONS  \ottnt{p}   \ottsym{)} \, e = v_{{\mathrm{2}}} }
\end{gather*}
  \caption{PEAK semantics for a program $P$.}
  \label{fig:peak}
\end{figure}

Figure~\ref{fig:peakfun} shows the syntax for the PEAK abstract machine,
as well as some helper functions described below. Unlike the CEK machine, the PEAK machine
works with paths in the AST rather than directly with
CBPV terms. It also uses an extra ``argument stack'' $a$ to
delay the creation of closures. A PEAK state $\rho$ is thus a quadruple $ \langle \ottnt{p}  ,\;  e  ,\;  a  ,\;  k \rangle $, where $e$ is the
environment, which maps paths to values, and $k$ is the stack, whose
frames mirror the CEK frames.

\paragraph{Path Environments.}
Earlier, we remarked that CEK machine closures
represent terms with delayed substitutions. Unloading carried out
this substitution explicitly to recover exactly the residual of
reduction using the SOS rules. It should therefore be possible to add
bindings to an environment in a closure without changing the execution
behavior or affecting the relationship with SOS states, as long as
they do not bind additional free variables of the term. The main
difficulty in sharing environments to attain an optimized evaluation
strategy is that the desired property of substitutions is sensitive to
the exact representation of environments and binders.

%% The two most popular representations of binders are names and
%% DeBruijn indices. Neither of these is particularly well suited to
%% sharing environments. In a DeBruijn representation, bindings are added to the front of the list
%% as execution proceeds under binders, changing the interpretation of
%% each index. Using a named representation of binders runs into
%% problems with capturing variable names; one could require that all
%% bound names be unique and distinct from the free variables of the
%% source term, but our intended target language already contains an
%% alternative solution.

We borrow an idea from SSA and use a representation of environments
that associates dynamic values with \textit{paths} into the source program.
This turns out to have a number of advantages: 1. it imposes no
conditions on the representation of binders in source terms, 2. we
automatically get a unique name for each occurrence of a binder, which
facilitates sharing environments, and 3. it results in a definition of
machine states nearly identical to that of our desired CFG machine. 

A PEAK machine \textit{path} is simply a sequence of natural numbers
interpreted relative to a CBPV term $P$.\footnote{Here we use the
  metavariable $P$ rather than $M$ to range over CBPV terms to
  emphasize the fact that this term is being treated by a PEAK
  machine.} The empty path $\cdot$ denotes the entire term $P$,
whereas a path of the form $ \ottmv{n}  \CONS  \ottnt{p} $ denotes the $n^{th}$ subterm
(ordered from left to right) of the term denoted by $\ottnt{p}$, indexed
from $0$. For example, for $\ottnt{P}= \kw{if0} \;  \ottnt{V} \;  \langle  \ottnt{M_{{\mathrm{1}}}}  \ottsym{,}  \ottnt{M_{{\mathrm{2}}}}  \rangle \;  \ottnt{M_{{\mathrm{3}}}} $ the path
$ \ottsym{2}  \CONS  \cdot $ denotes the subterm $\ottnt{M_{{\mathrm{3}}}}$. We write $ \ottnt{P}  [  \ottnt{p}  ] $ for
the subterm of $\ottnt{P}$ at path $\ottnt{p}$, and so for the term above we
also have $ \ottnt{P}  [   \ottsym{0}  \CONS  \cdot   ]  = \ottnt{V}$, $ \ottnt{P}  [   \ottsym{0}  \CONS   \ottsym{1}  \CONS  \cdot    ]  = \ottnt{M_{{\mathrm{1}}}}$,
and $ \ottnt{P}  [   \ottsym{1}  \CONS   \ottsym{1}  \CONS  \cdot    ]  = \ottnt{M_{{\mathrm{2}}}}$.  (Clearly $\ottnt{P}[-]$ is a
partial function.)

% In order to associate dynamic values with the location of a binder in
% the source term, the machine states keeps track of the path from the
% root of the term to the currently executing subterm. Since the control
% component of CEK states is always a subterm of the source program, we
% can redefine the transition relation to take the source program as a
% parameter to avoid maintaining redundant information.
%
%% TODO: go into detail about path and environment representation,
%% - either write down lookup fn or explain that lookup starts with
%%   tail of path ... 
%% envs are finite maps from paths to dynamic values
%% equivalence rel. =p, when maps agree for all prefixes of p, etc.
%

\paragraph{PEAK operational semantics.}

Figure~\ref{fig:peak} shows the operational semantics for the PEAK
abstract machine. The first component $\ottnt{p}$ of a PEAK state
$ \langle \ottnt{p}  ,\;  e  ,\;  a  ,\;  k \rangle $ is a path that tracks the location of the
currently executing subterm, relative to a CBPV term $\ottnt{P}$. Since
$\ottnt{P}$ is constant for the duration of the execution, the step
function of the PEAK machine is parameterized by it. %During execution,
%paths into subterms of $\ottnt{P}$ are dynamically manipulated as needed.
The column labeled ``when $ \ottnt{P}  [  \ottnt{p}  ] $ = " shows the part of the
subterm that needs to be consulted to determine which evaluation rule
to apply.  The indices on the \underline{\ \ } ``wildcard'' pattern mark the
numbers used to construct the paths.  For instance, if $ \ottnt{P}  [  \ottnt{p}  ]  = \kw{if0} \;   \underline{\ \ } _{ \ottsym{0} }  \;   \underline{\ \ } _{ \ottsym{1} }  \;   \underline{\ \ } _{ \ottsym{2} }  $, then the subterm of the guard is $ \ottsym{0}  \CONS  \ottnt{p} $.

Sharing environments requires that the machine recognize when it is
safe to delay creating a closure. PEAK machine configurations contain
an additional ``argument stack'' component that logically contains a
prefix of the continuation stack for which the current environment is
a valid substitution. The syntax of argument frames is presented in
Fig. \ref{fig:peak}. In the argument stack, $ \kwl{ARG} $, $ \kwl{SEQ} $ and
$ \kwl{PRJ} $ correspond to application, sequencing, and projection
continuation frames, respectively.

\begin{figure}[t]
\small
  \[
  \begin{array}{lcl@{\quad }l}
    \kw{unload}_e & : & \mathrm{Path} \to \mathrm{Env} \to \mathrm{Env_{CEK}} \\
     \kw{unload}_{e} ( \cdot ,\; e )  & = & \cdot \\
     \kw{unload}_{e} (  \ottmv{n}  \CONS  \ottnt{p}  ,\; e )  
      & = & \ottsym{[}  \ottmv{x}  \mathbin{\mapsto}   \kw{unload}_{v} ( v )   \ottsym{]}  \ottsym{(}   \kw{unload}_{e} ( \ottnt{p} ,\; e )   \ottsym{)}
        & \mbox{when $ \ottnt{P}  [  \ottnt{p}  ]  =  \lambda \ottmv{x} . \underline{\ \ } $, $n = 0$, and $e(\ottnt{p}) = v$ } \\ 
      & = & \ottsym{[}  \ottmv{x}  \mathbin{\mapsto}   \kw{unload}_{v} ( v )   \ottsym{]}  \ottsym{(}   \kw{unload}_{e} ( \ottnt{p} ,\; e )   \ottsym{)}
        & \mbox{when $ \ottnt{P}  [  \ottnt{p}  ]  = \underline{\ \ } \, \kw{to} \, \ottmv{x} \, \kw{in} \, \underline{\ \ }$, $n = 1$, and $e(\ottnt{p}) = v$ } \\ 
      & = & \ottsym{[}  \ottmv{x_{\ottmv{i}}}  \mathbin{\mapsto}  \ottnt{M_{\ottmv{i}}}  \ottsym{]}  \ottsym{(}   \kw{unload}_{e} ( \ottnt{p} ,\; e )   \ottsym{)}
        & \mbox{when $ \ottnt{P}  [  \ottnt{p}  ]  = \kw{letrec} \, \ottmv{x_{\ottmv{i}}}  \ottsym{=}  \ottnt{M_{\ottmv{i}}} \, \kw{in} \, \underline{\ \ }$, $n = 0$, and $e(\ottnt{p}) = v$ } \\ 
% SAZ: Note sure about this
      % & = & \texttt{\textcolor{red}{<<no parses (char 9):  [x := [r***ec x. lookup P 0:p, punloade p pe] ] (punloade p pe)  >>}}
      %   & \mbox{when $ \ottnt{P}  [  \ottnt{p}  ]  = \texttt{\textcolor{red}{<<no parses (char 1): r***ec x. \_ >>}}$, and $n = 0$ } \\
      & = &  \kw{unload}_{e} ( \ottnt{p} ,\; e )  & \mbox{otherwise} \\
  \end{array}
  \]
  \[
  \begin{array}{lcr@{\ }l}
    \kw{unload}_v & : & \mathrm{Val} \to \mathrm{Val_{CEK}} \\
     \kw{unload}_{v} ( \ottmv{x} )  & = & \ottmv{x} \\
     \kw{unload}_{v} ( \ottmv{n} )  & = & \ottmv{n} \\ 
     \kw{unload}_{v} ( \ottsym{[}  \ottnt{p}  \ottsym{,}  e  \ottsym{]} )  & = & \ottsym{[}   \ottnt{P}  [  \ottnt{p}  ]   \ottsym{,}   \kw{unload}_{e} ( \ottnt{p} ,\; e )   \ottsym{]}
  \end{array}
  \]
  \[
  \begin{array}{lcr@{\ }l}
    \kw{unload}_k & : & \mathrm{Env} \to \mathrm{Args} \to \mathrm{Kont} &
    \to \mathrm{Kont_{CEK}} \\
     \kw{unload}_k ( e ,\;  \cdot ,\;  \cdot )  & = & \cdot \\
     \kw{unload}_k ( e ,\;    \kwl{ARG}\; \ottnt{p}   ::  a  ,\;  k )  & = & 
       \kw{unload}_{v} ( \gamma \,  \ottsym{0}  \CONS  \ottnt{p}  \, e )   \mathbin{\cdot} \, \underline{\ \ } &::  \kw{unload}_k ( e ,\;  a ,\;  k )  \\ %TODO: decide to switch to PEAK_ALT?
     \kw{unload}_k ( e ,\;   \kwl{SEQ} \, \ottnt{p}  ::  a  ,\;  k )  & = & 
      \ottsym{[} \, \underline{\ \ } \, \kw{to} \, \ottmv{x} \, \kw{in} \,  \ottnt{P}  [   \ottsym{1}  \CONS  \ottnt{p}   ]   \ottsym{,}   \kw{unload}_{e} ( \ottnt{p} ,\; e )   \ottsym{]} &::  \kw{unload}_k ( e ,\;  a ,\;  k )  \\
     \kw{unload}_k ( e ,\;  \cdot ,\;   \ottsym{(}  v  \mathbin{\cdot} \, \underline{\ \ }  \ottsym{)} :: k  )        & = & 
       \kw{unload}_{v} ( v )   \mathbin{\cdot} \, \underline{\ \ } &::  \kw{unload}_k ( e ,\;  \cdot ,\;  k )  \\
     \kw{unload}_k ( e ,\;  \cdot ,\;   \ottsym{[}  \ottnt{p}  \ottsym{,}  e'  \ottsym{,}  a'  \ottsym{]} :: k  )  & = & 
      \ottsym{[} \, \underline{\ \ } \, \kw{to} \, \ottmv{x} \, \kw{in} \,  \ottnt{P}  [   \ottsym{1}  \CONS  \ottnt{p}   ]   \ottsym{,}   \kw{unload}_{e} ( \ottnt{p} ,\; e' )   \ottsym{]} &::  \kw{unload}_k ( e' ,\;  a' ,\;  k )  \\
  \end{array}
  \]
  \[
  \begin{array}{lcr@{\ }l}
    \kw{unload} & : & \mathrm{State_{PEAK}} \to
      \mathrm{State_{CEK}} \\
    \kw{unload}( \langle \ottnt{p}  ,\;  e  ,\;  a  ,\;  k \rangle ) & = &  \langle  \ottnt{P}  [  \ottnt{p}  ]   ,\;   \kw{unload}_{e} ( \ottnt{p} ,\; e )   ,\;   \kw{unload}_k ( e ,\;  a ,\;  k )  \rangle 
  \end{array}
  \]
  \caption{PEAK unloading for a program $P$.}
  \label{fig:pekunlfuns}
\end{figure}

\paragraph{The Argument Stack}
The PEAK machine rules for sequencing and application push frames onto the argument stack, batching them up.
The produce, lambda, and tuple rules each have two cases. When the
argument stack is empty, the machine executes roughly as the CEK
machine. The rule for $ \kw{force} $ now evaluates the argument stack up
until the nearest $ \kwl{SEQ} $ and pushes the resulting frames on the
continuation stack---this operation, written as $ \delta $, is shown in
Figure~\ref{fig:peakfun}. The remaining argument frames are also saved
in the sequence frame, and are restored when the sequence frame is
popped by $ \kw{prd} $. The rule for $ \kw{letrec} $ is similar to sequencing or application, but does not push an argument frame, since it has no associated continuation.

The PEAK version of the $ \gamma $ function
is adapted from the CEK's $ \gamma $ function to work on paths rather
than directly on values. Its definition is shown in
Figure~\ref{fig:peakfun}. Note that a closure is only constructed when $ \gamma $ is called on a $ \kw{thunk} $; this means that construction of a closure for a $ \kw{letrec} $ is delayed until the variable bound by a $ \kw{letrec} $ appears under a
$ \kw{prd} $ or is applied. If a variable bound by a $ \kw{letrec} $ is
forced, the closure returned by $ \gamma $ just passes through the
current environment.
% \steve{This argument stack is a bit like Leroy's ZINC machine, which
%   also uses an argument stack, right? We should probably cite that
%   here too.}

\paragraph{Unloading PEAK Machine States to CEK.}
Recovering CEK states from PEAK states involves converting dynamic
values, argument frames, and continuation frames. The $ \kw{PEAK.unload} $
definition is given in Figure~\ref{fig:pekunlfuns}. Recovering a term from a
path is a simple lookup in the source program, while unloading the
PEAK environment proceeds by recursion on the path component of the
closure. Every time a path corresponding to a binder is encountered,
the associated dynamic value is unloaded and added to the CEK
environment. The $ \kw{letrec} $ binder is a special case: instead of binding
a dynamic value, it binds a list of terms, which are stored directly in the CEK environment.
%% 
%% TODO: figure, unload dval

To produce a CEK continuation stack from a PEAK machine state, the
continuation stack is interleaved with argument frames. Starting with
the current argument stack, each argument frame is converted to a CEK
frame using the current environment. Then, the PEAK continuation stack
is traversed, converting application and projection frames as
necessary. When a sequence frame is encountered, the process is
iterated using the saved environment and argument stack and the
remainder of the continuation stack.
%% 
%% TODO: figure for unload_oframe, unwind

\paragraph{Well-Formed States.}
The representation of environments used by the PEAK machine requires
some extra well-formedness constraints on states in order to guarantee
that free variables will not be captured. 
\begin{definition}
  A PEAK state $ \langle \ottnt{p}  ,\;  e  ,\;  a  ,\;  k \rangle $ is \emph{well-formed} when both of the following hold:
  \begin{enumerate}
  \item In each closure as well as the path and environment of the
    machine state, every suffix of the path that corresponds to a
    binder has an associated dynamic value in the environment.
  \item The path in the first argument frame is a suffix of the current path $\ottnt{p}$;
	the path of each other argument frame is a suffix of the path of the preceding frame. 
%Every path occurring in an argument frame is a suffix of the
    %path of the preceding $ \kwl{SEQ} $ frame, or the current path of the
    %machine state or containing continuation frame.

  \end{enumerate}
\end{definition}
The first condition is a scoping property that rules out states in which 
bound variables of the source term do not have
associated dynamic values in the substitution. The second condition is
necessary to maintain this invariant when argument frames are pushed
onto and restored from the continuation stack.

As described in Section~\ref{sec:proof}, we must prove that CBPV terms load to well-formed states,
and that well-formedness is preserved by machine steps. Loading is simply the constant function that produces the initial state $ \langle \cdot  ,\;  \cdot  ,\;  \cdot  ,\;  \cdot \rangle $, so the first proof is simple; the second follows from the definition of the step function. 

\paragraph{Correctness with Respect to the CEK Machine}
The PEAK step and unloading functions, as well as the state
well-formedness predicate, are now defined with respect to an initial
CBPV term. The statement of correctness is adjusted from that of the CEK machine accordingly.

\begin{lemma} (PEAK Correctness) For all terms $\ottnt{P}$ and PEAK states $\rho$ that
  are well-formed with respect to $\ottnt{P}$, 
\[
\kw{CEK.step} \, \ottsym{(}  \kw{PEAK.unload} \, \ottnt{P} \, \rho  \ottsym{)} = 
\kw{PEAK.unload} \, \ottnt{P} \, \ottsym{(}  \kw{PEAK.step} \, \ottnt{P} \, \rho  \ottsym{)}
\]
\end{lemma}

\section{PEK Machine}
\label{sec:pek}

\begin{figure}
%% \footnotesize
%  
%% TODO: integrate into section; align some seps; all sorts of syntax changes
\[
\begin{array}[t]{r@{\ }l@{\ \ }c@{\ }lll}
  \mathrm{State} \ \ni\ & \rho & =\  & \mathrm{Path} \times \mathrm{Env}
    \times \mathrm{Kont} \\
  \mathrm{Path} \ \ni\ & \ottnt{p}, \ottnt{q} & \bnf\ & \cdot \sep  \ottmv{n}  \CONS  \ottnt{p}  \\
  \mathrm{Env} \ \ni\ & e & =\ & \mathrm{Path}
    \longrightarrow_{\mathrm{fin}} \mathrm{Val} \\
  \mathrm{Args} \ \ni\ & a & \bnf\ & \cdot \sep   \kwl{ARG}\; \ottnt{p}   ::  a 
    \sep  \kwl{SEQ} \, \ottnt{p}  ::  a  \\
  \mathrm{Kont} \ \ni\ & k & \bnf\ & \cdot \sep  v  \mathbin{\cdot} \, \underline{\ \ } :: k 
    \sep  \ottsym{[}  \ottnt{p}  \ottsym{,}  \ottnt{q}  \ottsym{,}  e  \ottsym{]} :: k  \\
  %% TODO: errors?
  \mathrm{Val} \ \ni\ & v & \bnf\ & \ottmv{x} \sep \ottmv{n} \sep \ottsym{[}  \ottnt{p}  \ottsym{,}  e  \ottsym{]}
\end{array}
\]

\begin{align*}
   \epsilon & : & \mathrm{Path} \rightarrow \mathrm{Path}\\
   \mathrm{aframes} & : & \mathrm{Path} \rightarrow \mathrm{Args}
\end{align*}

\[
\begin{array}{rcll}
  \gamma & : & \mathrm{Path} \rightarrow \mathrm{Env} \rightarrow \mathrm{DVal} \\
  \gamma \, \ottnt{p} \, e & = & \mathrm{lookup}(\ottnt{P}, \ottnt{p}, \ottnt{e}, \ottmv{x})
                        & \mbox{when $ \ottnt{P}  [  \ottnt{p}  ]  = \ottmv{x}$} \\
  \gamma \, \ottnt{p} \, e & = & \ottsym{[}   \eta\;   \ottsym{0}  \CONS  \ottnt{p}    \ottsym{,}  e  \ottsym{]} 
                      & \mbox{when $ \ottnt{P}  [  \ottnt{p}  ]  = \kw{thunk} \, \underline{\ \ }$} \\
  \gamma \, \ottnt{p} \, e & = & \ottmv{n} 
                      & \mbox{when $ \ottnt{P}  [  \ottnt{p}  ]  = \ottmv{n}$} \\

\end{array}
\]

\[
  \begin{array}{rcll}
    \delta & : & \mathrm{Env} \rightarrow \mathrm{Args} \rightarrow \mathrm{Kont} \\
     \delta\;  e  \;  \cdot  & = & . \\
     \delta\;  e  \;  \ottsym{(}    \kwl{ARG}\; \ottnt{p}   ::  a   \ottsym{)}  & = &  \ottsym{(}  \gamma \, \ottsym{(}   \ottsym{0}  \CONS  \ottnt{p}   \ottsym{)} \, e  \mathbin{\cdot} \, \underline{\ \ }  \ottsym{)} ::  \delta\;  e  \;  a   \\
     \delta\;  e  \;  \ottsym{(}   \kwl{SEQ} \, \ottnt{p}  ::  a   \ottsym{)}  & = &  \ottsym{[}  \ottnt{p}  \ottsym{,}   \eta\;   \ottsym{1}  \CONS  \ottnt{p}    \ottsym{,}  e  \ottsym{]} :: \cdot  \\
\end{array}
\]

\caption{PEK machine states and value-semantic functions}
\label{fig:pekfun}
\end{figure}

\begin{figure}
  \begin{minipage}{0.5\linewidth}
\begin{lstlisting}[basicstyle={\sffamily},mathescape=true]
$\kw{aframes} \, \ottsym{(}   \ottmv{n}  \CONS  \ottnt{p}   \ottsym{)}$ =
  match $\ottmv{n}$, $ \ottnt{P}  [   \ottmv{n}  \CONS  \ottnt{p}   ] $ with
   | 1, $  \underline{\ \ } _{ \ottsym{0} }  \!\cdot\!  \underline{\ \ } _{ \ottsym{1} }  $ $\TO$ $  \kwl{ARG}\; \ottnt{p}   ::  \kw{aframes} \, \ottnt{p} $
   | 0, $ \lambda \ottmv{x} .  \underline{\ \ } _{ \ottsym{0} }  $ $\TO$ match $\kw{aframes} \, \ottnt{p}$ with 
                | $  \kwl{ARG}\; \ottnt{q}   ::  a $ $\TO$ $a$
                | $a$ $\TO$ $a$
                end
   | 0, $ \underline{\ \ } _{ \ottsym{0} }  \, \kw{to} \, \ottmv{x} \, \kw{in} \,  \underline{\ \ } _{ \ottsym{1} } $ $\TO$ $ \kwl{SEQ} \, \ottnt{p}  ::  \kw{aframes} \, \ottnt{p} $
   | 0, $\kw{letrec} \, \ottmv{x_{\ottmv{i}}}  \ottsym{=}   \underline{\ \ } _{ \ottmv{i} }  \, \kw{in} \,  \underline{\ \ } _{ \ottsym{0} } $ | 1, $ \underline{\ \ } _{ \ottsym{0} }  \, \kw{to} \, \ottmv{x} \, \kw{in} \,  \underline{\ \ } _{ \ottsym{1} } $ | 1, $ \kw{if0} \;   \underline{\ \ } _{ \ottsym{0} }  \;   \underline{\ \ } _{ \ottsym{1} }  \;   \underline{\ \ } _{ \ottsym{2} }  $ 
   | 2, $ \kw{if0} \;   \underline{\ \ } _{ \ottsym{0} }  \;   \underline{\ \ } _{ \ottsym{1} }  \;   \underline{\ \ } _{ \ottsym{2} }  $ $\TO$ $\kw{aframes} \, \ottnt{p}$
   | _ $\TO$ $\cdot$
  end
\end{lstlisting}
  \end{minipage}
\caption{Calculating PEK argument frames}
\label{fig:pekaframes}
\end{figure}

\begin{figure}[ht]
  %% \scriptsize

\begin{gather*}
\infer
    {  \langle \ottnt{p}  ,\;  e  ,\;  k \rangle 
       \longrightarrow 
       \langle \ottnt{p'}  ,\;  e'  ,\;   \delta\;  e \;  a   +\!\!\!+  k \rangle  }
    {  \ottnt{P}  [  \ottnt{p}  ]       = \kw{force} \,  \underline{\ \ } _{ \ottsym{0} }  &
       \mathrm{aframes}( \ottnt{p} )        = a &
      \gamma \, \ottsym{(}   \ottsym{0}  \CONS  \ottnt{p}   \ottsym{)} \, e = \ottsym{[}  \ottnt{p'}  \ottsym{,}  e'  \ottsym{]} }
\\
\infer
    {  \langle \ottnt{p}  ,\;  e  ,\;  k \rangle 
       \longrightarrow 
       \langle  \eta\;   \ottsym{1}  \CONS  \ottnt{p'}    ,\;  e  \ottsym{[}  \ottnt{p'}  \mathbin{\mapsto}  v  \ottsym{]}  ,\;  k \rangle  }
    {  \ottnt{P}  [  \ottnt{p}  ]       =  \kw{prd} \,  \underline{\ \ } _{ \ottsym{0} }  &
       \mathrm{aframes}( \ottnt{p} )        =   \kwl{SEQ} \, \ottnt{p'}  ::  a  &
      \gamma \, \ottsym{(}   \ottsym{0}  \CONS  \ottnt{p}   \ottsym{)} \, e =  v }
\\
\infer
    {  \langle \ottnt{p}  ,\;  e  ,\;   \ottsym{[}  \ottnt{q}  \ottsym{,}  \ottnt{p'}  \ottsym{,}  e'  \ottsym{]} :: k  \rangle 
       \longrightarrow 
       \langle \ottnt{p'}  ,\;  e'  \ottsym{[}  \ottnt{q}  \mathbin{\mapsto}  v  \ottsym{]}  ,\;  k \rangle  }
    {  \ottnt{P}  [  \ottnt{p}  ]       = \kw{prd} \,  \underline{\ \ } _{ \ottsym{0} }  &
       \mathrm{aframes}( \ottnt{p} )        = \cdot &
      \gamma \, \ottsym{(}   \ottsym{0}  \CONS  \ottnt{p}   \ottsym{)} \, e = v }
\\
\infer
    {  \langle \ottnt{p}  ,\;  e  ,\;  k \rangle 
       \longrightarrow 
       \langle  \eta\;   \ottsym{0}  \CONS  \ottnt{p}    ,\;  e  \ottsym{[}  \ottnt{p}  \mathbin{\mapsto}  v  \ottsym{]}  ,\;  k \rangle  }
    {  \ottnt{P}  [  \ottnt{p}  ]       =  \lambda \ottmv{x} .  \underline{\ \ } _{ \ottsym{0} }   & 
       \mathrm{aframes}( \ottnt{p} )        =   \kwl{ARG}\; \ottnt{q}   ::  a  &
      \gamma \, \ottsym{(}   \ottsym{0}  \CONS  \ottnt{q}   \ottsym{)} \, e = v }
\\
\infer
    {  \langle \ottnt{p}  ,\;  e  ,\;   v  \mathbin{\cdot} \, \underline{\ \ } :: k  \rangle 
       \longrightarrow 
       \langle  \eta\;   \ottsym{0}  \CONS  \ottnt{p}    ,\;  e  \ottsym{[}  \ottnt{p}  \mathbin{\mapsto}  v  \ottsym{]}  ,\;  k \rangle  }
    {  \ottnt{P}  [  \ottnt{p}  ]       =  \lambda \ottmv{x} .  \underline{\ \ } _{ \ottsym{0} }   &
       \mathrm{aframes}( \ottnt{p} )        = \cdot }
\\
\infer
    {  \langle \ottnt{p}  ,\;  e  ,\;  k \rangle 
       \longrightarrow 
       \langle  \eta\;   \ottsym{1}  \CONS  \ottnt{p'}    ,\;  e  \ottsym{[}  \ottnt{p'}  \mathbin{\mapsto}  v_{{\mathrm{1}}}  \means{\oplus}  v_{{\mathrm{2}}}  \ottsym{]}  ,\;  k \rangle  }
    {  \ottnt{P}  [  \ottnt{p}  ]       =  \underline{\ \ } _{ \ottsym{0} }   \oplus   \underline{\ \ } _{ \ottsym{1} }  &
       \mathrm{aframes}( \ottnt{p} )        =  \kwl{SEQ} \, \ottnt{p'}  ::  a  &
      \gamma \, \ottsym{(}   \ottsym{0}  \CONS  \ottnt{p}   \ottsym{)} \, e = v_{{\mathrm{1}}} & 
      \gamma \, \ottsym{(}   \ottsym{1}  \CONS  \ottnt{p}   \ottsym{)} \, e = v_{{\mathrm{2}}} }
\\
\infer
    {  \langle \ottnt{p}  ,\;  e  ,\;   \ottsym{[}  \ottnt{q}  \ottsym{,}  \ottnt{p'}  \ottsym{,}  e'  \ottsym{]} :: k  \rangle 
       \longrightarrow 
       \langle \ottnt{p'}  ,\;  e'  \ottsym{[}  \ottnt{q}  \mathbin{\mapsto}  v_{{\mathrm{1}}}  \means{\oplus}  v_{{\mathrm{2}}}  \ottsym{]}  ,\;  k \rangle  }
    {  \ottnt{P}  [  \ottnt{p}  ]       =  \underline{\ \ } _{ \ottsym{0} }   \oplus   \underline{\ \ } _{ \ottsym{1} }  &
       \mathrm{aframes}( \ottnt{p} )        = \cdot &
      \gamma \, \ottsym{(}   \ottsym{0}  \CONS  \ottnt{p}   \ottsym{)} \, e = v_{{\mathrm{1}}} & 
      \gamma \, \ottsym{(}   \ottsym{1}  \CONS  \ottnt{p}   \ottsym{)} \, e = v_{{\mathrm{2}}} }
\end{gather*}

  \caption{PEK semantics for a program $P$ }
  \label{fig:pek}
\end{figure}

Examining the transitions of the PEAK machine, it is clear that the search
rules for sequencing and application now do very little work: they merely
save a path on the argument stack. Similarly, since bindings made by
$ \kw{letrec} $ can be determined from the path and initial program, the
associated transition rule just advances the program counter. Since the
argument stack is pushed onto the continuation stack every time execution
reaches a thunk, it turns out to be possible to compute the state of the
argument stack at any program point entirely statically.

To achieve our final execution strategy, we push the manipulation of
the argument stack into an auxiliary function, $ \kw{aframes} $, shown in
Fig. \ref{fig:pekaframes}. In the resulting transition system the
search rules become trivial, merely advancing the program counter to
the next redex before it is reduced. We replace these rules with a
simple function, $\eta$, that is used to calculate the next path in
each transition rule. The final transition system is shown in Fig.
\ref{fig:pek}.

\paragraph{Unloading, well-formedness of states, and correctness.}
Unloading PEK machine states to PEAK states is straightforward: the
representation of dynamic values and environments is identical modulo the $\eta$ operator, and
the machines traverse terms in lockstep. Application
frames in the continuation stack are also unchanged. For sequencing
frames, we only have to throw away the extra path to the next
instruction, and recover the stored argument stack. The argument stack
in a sequence frame $\ottsym{[}  \ottnt{p}  \ottsym{,}  e  \ottsym{,}  a  \ottsym{]}$ of the PEAK machine is
always the one grabbed at path $\ottnt{p}$, so we can simply reuse the
$ \kw{aframes} $ function used in the semantics. Similarly, we can
recover the argument stack of the top-level PEAK machine state. The
well-formedness condition for PEK machine states is also simplified.
We require only that paths point to values or computations as
necesssary, and environments in the machine state, closures, and the
continuations stack contain values for all binders in scope of the
associated path. Loading, again, always produces the initial state $ \langle \cdot  ,\;  \cdot  ,\;  \cdot \rangle $.

\begin{lemma} (PEK Correctness) For all terms $\ottnt{P}$ and PEK states $\rho$ that
  are well-formed with respect to $\ottnt{P}$, 
\[
\kw{PEAK.step} \, \ottnt{P} \, \ottsym{(}  \kw{PEK.unload} \, \ottnt{P} \, s  \ottsym{)} = 
\kw{PEK.unload} \, \ottnt{P} \, \ottsym{(}  \kw{PEK.step} \, \ottnt{P} \, s  \ottsym{)}
\]
\end{lemma}

\section{CFG Machine}
\label{sec:cfg}
From the PEK machine, we now derive what Danvy~\cite{danvyvirtual} calls a \emph{virtual machine}: a machine that operates on a structured program with an instruction set, rather than a lambda term.
The abstract machine is factored into a compiler and interpreter for
an instruction set that reflects the cases of the step function.
The compiler pre-computes various aspects of execution that can be
determined statically, so that the interpreter can be simpler and closer to real assembly languages.

\begin{figure}
\vspace{-5mm}
%  
%% TODO: integrate into section; align some seps; all sorts of syntax changes
\[
\begin{array}[t]{r@{\ }l@{\ \ }c@{\ }lll}
  \mathrm{State} \ \ni\ & \rho & =\  & \mathrm{Path} \times \mathrm{Env}
    \times \mathrm{Kont} \\
  \mathrm{Ident} \ \ni\ & \ottnt{p}, \ottnt{q} \\%& \bnf\ & \cdot \sep  \ottmv{n}  \CONS  \ottnt{p}  \\
  \mathrm{Env} \ \ni\ & e & =\ & \mathrm{Path}
    \longrightarrow_{\mathrm{fin}} \mathrm{Val} \\
  \mathrm{Args} \ \ni\ & a & \bnf\ & \cdot \sep   \kwl{ARG}\; \ottnt{p}   ::  a 
    \sep  \kwl{SEQ} \, \ottnt{p}  ::  a  \\
  \mathrm{Kont} \ \ni\ & k & \bnf\ & \cdot \sep  v  \mathbin{\cdot} \, \underline{\ \ } :: k 
    \sep  \ottsym{[}  \ottnt{p}  \ottsym{,}  \ottnt{q}  \ottsym{,}  e  \ottsym{]} :: k  \\
  %% TODO: errors?
  \mathrm{Val} \ \ni\ & v & \bnf\ & \ottmv{x} \sep \ottmv{n} \sep \ottsym{[}  \ottnt{p}  \ottsym{,}  e  \ottsym{]}
\end{array}
\]

\[
\begin{array}[t]{r@{ }l@{\ }r@{\  }ll@{}l}
  \mathrm{Operands} \ \ni\ & o & \bnf\ & \multicolumn{2}{l}{\kwl{VAR} \, \ottmv{x} 
    \sep \kwl{NAT} \, \ottmv{n} \sep \kwl{LOC} \, \ottnt{p} \sep \kwl{LBL} \, \ottnt{p}} \\

  \mathrm{Instructions} \ \ni\ & \ottnt{ins} & \bnf\ 
        & \kwl{CALL} \, o \, \ottsym{(}  \overline{o}  \ottsym{)} \, \ottnt{p} \\
  && |\ & \kwl{TAIL} \, o \, \ottsym{(}  \overline{o}  \ottsym{)} \\
  && |\ & \kwl{MOV} \, o \, \ottnt{p} \sep \kwl{RET} \, o \sep \kwl{POP} \, \ottnt{p} \\
  && |\ & \kwl{IF0} \, o \, \ottnt{p_{{\mathrm{1}}}} \, \ottnt{p_{{\mathrm{2}}}} \sep  \kwl{OP}\; o_{{\mathrm{1}}}   \oplus   o_{{\mathrm{2}}}  \;  \ottnt{p}  \sep  \kwl{OPRET}\;  o_{{\mathrm{1}}}   \oplus   o_{{\mathrm{2}}} \\
\end{array}
\]

%% \[
%% \begin{array}{rcll}
%%   \gamma & : & \mathrm{Env} \rightarrow \mathrm{Operand} \rightarrow \mathrm{DVal} \\
  %% \gamma \, \ottnt{p} \, e & = & \mathrm{lookup}(\ottnt{P}, \ottnt{p}, \ottnt{e}, \ottmv{x})
  %%                       & \mbox{when $ \ottnt{P}  [  \ottnt{p}  ]  = \ottmv{x}$}
  %% \\
  %% \gamma \, \ottnt{p} \, e & = & \ottsym{[}   \eta\;   \ottsym{0}  \CONS  \ottnt{p}    \ottsym{,}  e  \ottsym{]} 
  %%                     & \mbox{when $ \ottnt{P}  [  \ottnt{p}  ]  = \kw{thunk} \, \underline{\ \ }$}
  %% \\
  %% \delta & : & \mathrm{Env} \rightarrow \mathrm{[Operand]} \rightarrow \mathrm{[Kont]} \\
  %%  \delta\;  e  \;  \cdot  & = & . \\
  %%  \delta\;  e  \;  \ottsym{(}    \kwl{ARG}\; \ottnt{p}   ::  a   \ottsym{)}  & = &  \ottsym{(}  \gamma \, \ottsym{(}   \ottsym{0}  \CONS  \ottnt{p}   \ottsym{)} \, e  \mathbin{\cdot} \, \underline{\ \ }  \ottsym{)} ::  \delta\;  e  \;  a   \\
  %%  \delta\;  e  \;  \ottsym{(}   \kwl{SEQ} \, \ottnt{p}  ::  a   \ottsym{)}  & = &  \ottsym{[}  \ottnt{p}  \ottsym{,}   \eta\;   \ottsym{1}  \CONS  \ottnt{p}    \ottsym{,}  e  \ottsym{]} :: \cdot  \\
%% \end{array}
%% \]

\caption{CFG machine states}
\label{fig:cfgfun}
\end{figure}

\begin{figure}[ht]

%% \vspace{-20mm}

\[
\hspace*{-0.9cm}
\begin{array}[t]{r@{\ \  \longrightarrow  \ \ }l@{\quad}lll}
  \mbox{start state} & \mbox{next state} 
  & \mbox{$\mathrm{Compile}(P,p)=$ }
  & \mbox{and} \\
  %% & \mbox{$\gamma \, \ottsym{(}   \ottsym{0}  \CONS  \ottnt{p}   \ottsym{)} \, e =$}  \\
    \hline

   \langle \ottnt{p}  ,\;  e  ,\;  k \rangle 
  &  \langle \ottnt{p'}  ,\;  e'  ,\;   \delta\;  e \;  \overline{o}   +\!\!\!+  k \rangle 
  & \kwl{TAIL} \, o \, \overline{o},\cdot
  &  \kw{eval} \;  e \;  o  = \ottsym{[}  \ottnt{p'}  \ottsym{,}  e'  \ottsym{]}
  \\[1mm]

   \langle \ottnt{p}  ,\;  e  ,\;  k \rangle 
  &  \langle \ottnt{p'}  ,\;  e'  ,\;   \delta\;  e \;  \overline{o}   +\!\!\!+   \ottsym{[}  \ottnt{q}  \ottsym{,}  \ottnt{p''}  \ottsym{,}  e'  \ottsym{]} :: k  \rangle 
  & \kwl{CALL} \, o \, \overline{o} \, \ottnt{q},[ \ottnt{p''} ]
  &  \kw{eval} \;  e \;  o  = \ottsym{[}  \ottnt{p'}  \ottsym{,}  e'  \ottsym{]}
  \\[1mm]

   \langle \ottnt{p}  ,\;  e  ,\;  k \rangle 
  &  \langle \ottnt{p'}  ,\;  e  \ottsym{[}  \ottnt{q}  \mathbin{\mapsto}  v  \ottsym{]}  ,\;  k \rangle 
  & \kwl{MOV} \, o \, \ottnt{q},[ \ottnt{p'} ]
  &  \kw{eval} \;  e \;  o  = v
  \\[1mm]

   \langle \ottnt{p}  ,\;  e  ,\;   \ottsym{[}  \ottnt{q}  \ottsym{,}  \ottnt{p'}  \ottsym{,}  e'  \ottsym{]} :: k  \rangle 
  &  \langle \ottnt{p'}  ,\;  e'  \ottsym{[}  \ottnt{q}  \mathbin{\mapsto}  v  \ottsym{]}  ,\;  k \rangle 
  & \kwl{RET} \, o,\cdot
  &  \kw{eval} \;  e \;  o  = v
  \\[1mm]

   \langle \ottnt{p}  ,\;  e  ,\;   v  \mathbin{\cdot} \, \underline{\ \ } :: k  \rangle 
  &  \langle \ottnt{p'}  ,\;  e  \ottsym{[}  \ottnt{q}  \mathbin{\mapsto}  v  \ottsym{]}  ,\;  k \rangle 
  & \kwl{POP} \, \ottnt{q},[ \ottnt{p'} ]
  \\[1mm]

   \langle \ottnt{p}  ,\;  e  ,\;  k \rangle 
  &  \langle \ottnt{p_{{\mathrm{1}}}}  ,\;  e  ,\;  k \rangle 
  & \kwl{IF0} \, o \, \ottnt{p_{{\mathrm{1}}}} \, \ottnt{p_{{\mathrm{2}}}},\cdot
  &  \kw{eval} \;  e \;  o  = \ottsym{0}
  \\[1mm]

   \langle \ottnt{p}  ,\;  e  ,\;  k \rangle 
  &  \langle \ottnt{p_{{\mathrm{2}}}}  ,\;  e  ,\;  k \rangle 
  & \kwl{IF0} \, o \, \ottnt{p_{{\mathrm{1}}}} \, \ottnt{p_{{\mathrm{2}}}},\cdot
  &  \kw{eval} \;  e \;  o   \neq  \ottsym{0}
  \\[1mm]

   \langle \ottnt{p}  ,\;  e  ,\;  k \rangle 
  &  \langle \ottnt{p'}  ,\;  e  \ottsym{[}  \ottnt{q}  \mathbin{\mapsto}  v_{{\mathrm{1}}}  \means{\oplus}  v_{{\mathrm{2}}}  \ottsym{]}  ,\;  k \rangle 
  &  \kwl{OP}\; o_{{\mathrm{1}}}   \oplus   o_{{\mathrm{2}}}  \;  \ottnt{q} ,[ \ottnt{p'} ]
  &  \kw{eval} \;  e \;  o_{{\mathrm{1}}}  = v_{{\mathrm{1}}},  \kw{eval} \;  e \;  o_{{\mathrm{2}}}  = v_{{\mathrm{2}}}
  \\[1mm]

   \langle \ottnt{p}  ,\;  e  ,\;   \ottsym{[}  \ottnt{q}  \ottsym{,}  \ottnt{p'}  \ottsym{,}  e'  \ottsym{]} :: k  \rangle 
  &  \langle \ottnt{p'}  ,\;  e'  \ottsym{[}  \ottnt{q}  \mathbin{\mapsto}  v_{{\mathrm{1}}}  \means{\oplus}  v_{{\mathrm{2}}}  \ottsym{]}  ,\;  k \rangle 
  &  \kwl{OPRET}\;  o_{{\mathrm{1}}}   \oplus   o_{{\mathrm{2}}} ,\cdot
  &  \kw{eval} \;  e \;  o_{{\mathrm{1}}}  = v_{{\mathrm{1}}},  \kw{eval} \;  e \;  o_{{\mathrm{2}}}  = v_{{\mathrm{2}}}
  
\end{array}
\]

\[
\begin{array}{crcll}
  &  \kw{eval}  : \mathrm{Env} \rightarrow \mathrm{Operand} & \rightarrow & \mathrm{DVal} \\
  &  \kw{eval} \;  e \;  \ottsym{(}  \kwl{VAR} \, \ottmv{x}  \ottsym{)}  & = & \ottmv{x} \\
  &  \kw{eval} \;  e \;  \ottsym{(}  \kwl{NAT} \, \ottmv{n}  \ottsym{)}  & = & \ottmv{n} \\
  &  \kw{eval} \;  e \;  \ottsym{(}  \kwl{LOC} \, \ottnt{p}  \ottsym{)}  & = & e  \ottsym{(}  \ottnt{p}  \ottsym{)} \\
  &  \kw{eval} \;  e \;  \ottsym{(}  \kwl{LBL} \, \ottnt{p}  \ottsym{)}  & = & \ottsym{[}  \ottnt{p}  \ottsym{,}  e  \ottsym{]} \\[2.5mm]
  & \delta : \mathrm{Env} \rightarrow \mathrm{[Operand]} & \rightarrow & \mathrm{[Kont]} \\
  &  \delta\;  e \;  \cdot  & = & . \\
  &  \delta\;  e \;  \ottsym{(}   o :: \overline{o}   \ottsym{)}  & = &  \ottsym{(}   \kw{eval} \;  e \;  o   \mathbin{\cdot} \, \underline{\ \ }  \ottsym{)} ::  \delta\;  e \;  \overline{o}  
\end{array}
\]

  \caption{CFG semantics for a program $P$ }
  \label{fig:cfg}
\end{figure}

When examining values, all but one case of the $ \gamma $ function does
not examine the environment. To do as much work as possible statically
in the compiler, we split the PEK machine's $ \gamma $ into two
functions, $ \overline{\gamma} $ and $\kw{eval}$. The former takes a
path to a syntactic value and computes an ``operand'', which may be a
free variable $\kwl{VAR} \, \ottmv{x}$, literal number $\kwl{NAT} \, \ottmv{n}$, a bound local
variable $\kwl{LOC} \, \ottnt{p}$ represented by the path to its binder, or a label
$\kwl{LBL} \, \ottnt{p}$ representing a position in the source program. The $ \kw{eval} $ function only has to look up bound
variables and attach an environment to labels in order produce dynamic
values from operands. It is easy to verify that
$\gamma \, \ottnt{p} \, e =  \kw{eval} \;  e \;  \overline{\gamma} \, \ottnt{p} $.

\paragraph{The CFG virtual machine.}
The target machine operates on a simple instruction set, whose 
operational semantics are shown in Figure~\ref{fig:cfg}. At
the CFG level, each program path $ \ottnt{P}  [  \ottnt{p}  ] $ now corresponds to
an instruction that communicates information about the source term and
argument stack at $\ottnt{p}$ to the virtual machine's transition function.
The instruction set contains 8 instructions:
\begin{itemize}
%\item $ \kwl{NOP} $, which does nothing
\item $ \kwl{CALL} $ and $ \kwl{TAIL} $, which perform regular and tail function calls (tail calls do not add return frames to the stack)
\item $ \kwl{MOV} $, which assigns a value to a variable
\item $ \kwl{OP} $, which performs an arithmetic operation and stores the result in a variable
\item $ \kwl{RET} $ and $ \kwl{ORET} $, which pop return frames and return a direct value and a computed value respectively
\item $ \kwl{POP} $, which pops arguments pushed by $ \kwl{CALL} $
\item $ \kwl{IF0} $, which performs a conditional branch
%\item $\kwl{SWI}  \,$, which chooses the next label from a list
\end{itemize}

% using the two previously defined functions, $ \overline{\gamma} $ and
% type $\mathrm{Ins}$ of instruction is shown in Fig. ?, and serves to

The states are the same as those of the PEK machine. The step function also corresponds closely to that of the PEK
machine, where the case analysis on $ \kw{aframes} $ is instead
handled by specialized instructions. Some transitions that were
distinguished in the PEK machine are represented by a single
instruction: for example, both lambdas and $ \kw{prd} $ are
represented by $ \kwl{MOV} $ instructions in cases where the enclosing frame is known. The step for $ \kw{force} $, on the other hand, has been
split into tail and non-tail calls.

Note that $ \kwl{SEQ} $ frames now contain two paths: the path bound to
the returned value in the environment, and the path of the next
instruction to execute. Though these will always be $\ottnt{p}$ and
$ \ottsym{1}  \CONS  \ottnt{p} $, respectively, this change allows us to treat paths
abstractly. The CFG machine never inspects the structure of paths or
creates new paths other than the ones already present in instructions.
Paths are only compared for equality by the environment and
instruction lookup functions, and thus could be replaced with abstract
identifiers.

\paragraph{Producing control flow graphs.} 
The output of our compiler is a mapping from paths to instructions. Instructions can
therefore reference positions in the CFG directly via ``labels'', a
common feature of real-world virtual machines. From
\[ 
 \kw{PEAK.step}  : \mathrm{Term} \to  \kw{PEAK.state}  \to  \kw{PEAK.state}  
\]
we obtain
\begin{align*}
 \kw{compile}  &: \mathrm{Term} \to \mathrm{Path} \to \mathrm{Ins} \times \mathrm{Path}^+ \\
 \kw{CFG.step}  &: (\mathrm{Path} \to \mathrm{Ins} \times \mathrm{Path}^+) \to
                       \kw{CFG.state}  \to  \kw{CFG.state} 
\end{align*}
where a CFG associates each program point with a pair of an instruction and a list of next program points.

\begin{figure*}[t]
  \begin{minipage}{0.6\linewidth}
  \[
  \begin{array}[t]{rlll}
    & \kw{compile} \, \ottnt{p} = 
    & \mbox{when $ \ottnt{P}  [  \ottnt{p}  ]  =$}  
    & \mbox{and when:}  \\
      \hline

   \\
     & \kwl{TAIL} \, \ottsym{(}  \overline{\gamma} \, \ottsym{(}   \ottsym{0}  \CONS  \ottnt{p}   \ottsym{)}  \ottsym{)} \, \overline{o}, \cdot & \kw{force} \,  \underline{\ \ } _{ \ottsym{0} }  & \kw{args} \, \ottnt{p} = (\overline{o},\kw{none})

   \\
     & \kwl{CALL} \, \ottsym{(}  \overline{\gamma} \, \ottsym{(}   \ottsym{0}  \CONS  \ottnt{p}   \ottsym{)}  \ottsym{)} \, \overline{o} \, \ottnt{q}, [  \eta\;   \ottsym{1}  \CONS  \ottnt{q}   ] & \kw{force} \,  \underline{\ \ } _{ \ottsym{0} }  & \kw{args} \, \ottnt{p}
                                                          = (\overline{o}, \kw{some}\;\ottnt{q})

   \\
     & \kwl{IF0} \, \ottsym{(}  \overline{\gamma} \, \ottsym{(}   \ottsym{0}  \CONS  \ottnt{p}   \ottsym{)}  \ottsym{)} \, \ottsym{(}   \eta\;   \ottsym{1}  \CONS  \ottnt{p}    \ottsym{)} \, \ottsym{(}   \eta\;   \ottsym{2}  \CONS  \ottnt{p}    \ottsym{)}, \cdot &  \kw{if0} \;   \underline{\ \ } _{ \ottsym{0} }  \;   \underline{\ \ } _{ \ottsym{1} }  \;   \underline{\ \ } _{ \ottsym{2} }   & 

   \\
     & \kwl{MOV} \, \ottsym{(}  \overline{\gamma} \, \ottsym{(}   \ottsym{0}  \CONS  \ottnt{p}   \ottsym{)}  \ottsym{)} \, \ottnt{q}, [  \eta\;  \ottsym{(}   \ottsym{1}  \CONS  \ottnt{q}   \ottsym{)}  ] & \kw{prd} \,  \underline{\ \ } _{ \ottsym{0} }  & \kw{aframes} \, \ottnt{p}=  \kwl{SEQ} \, \ottnt{q}  ::  a 

   \\
     & \kwl{RET} \, \ottsym{(}  \overline{\gamma} \, \ottsym{(}   \ottsym{0}  \CONS  \ottnt{p}   \ottsym{)}  \ottsym{)}, \cdot & \kw{prd} \,  \underline{\ \ } _{ \ottsym{0} }  & \kw{aframes} \, \ottnt{p} = \cdot

   \\
     & \kwl{MOV} \, \ottsym{(}  \overline{\gamma} \, \ottsym{(}   \ottsym{0}  \CONS  \ottnt{p'}   \ottsym{)}  \ottsym{)} \, \ottnt{p}, [  \eta\;   \ottsym{0}  \CONS  \ottnt{p}   ] &  \lambda \ottmv{x} .  \underline{\ \ } _{ \ottsym{0} }   & \kw{aframes} \, \ottnt{p} =   \kwl{ARG}\; \ottnt{p'}   ::  a 

   \\
     & \kwl{POP} \, \ottnt{p}, [  \eta\;  \ottsym{(}   \ottsym{0}  \CONS  \ottnt{p}   \ottsym{)}  ]&  \lambda \ottmv{x} .  \underline{\ \ } _{ \ottsym{0} }   & \kw{aframes} \, \ottnt{p} = \cdot

   \\
     &  \kwl{OP}\; \ottsym{(}  \overline{\gamma} \, \ottsym{(}   \ottsym{0}  \CONS  \ottnt{p}   \ottsym{)}  \ottsym{)}   \oplus   \ottsym{(}  \overline{\gamma} \,  \ottsym{1}  \CONS  \ottnt{p}   \ottsym{)}  \;  \ottnt{q} , [  \eta\;   \ottsym{1}  \CONS  \ottnt{q}   ] &  \underline{\ \ } _{ \ottsym{0} }   \oplus   \underline{\ \ } _{ \ottsym{1} }  & \kw{aframes} \, \ottnt{p} =  \kwl{SEQ} \, \ottnt{q}  ::  a 

   \\
     &  \kwl{OPRET}\;  \ottsym{(}  \overline{\gamma} \, \ottsym{(}   \ottsym{0}  \CONS  \ottnt{p}   \ottsym{)}  \ottsym{)}   \oplus   \ottsym{(}  \overline{\gamma} \,  \ottsym{1}  \CONS  \ottnt{p}   \ottsym{)} , \cdot &  \underline{\ \ } _{ \ottsym{0} }   \oplus   \underline{\ \ } _{ \ottsym{1} }  & \kw{aframes} \, \ottnt{p} = \cdot

  \end{array}
  \]
  \end{minipage}

\[  \begin{array}{rcll}
     \overline{\gamma}  & : & \mathrm{Path} \rightarrow \mathrm{Operand} \\
    \overline{\gamma} \, \ottnt{p} & = & \mathrm{\overline{lookup}}(\ottnt{P}, \ottnt{p}, \ottmv{x})
                          & \mbox{when $ \ottnt{P}  [  \ottnt{p}  ]  = \ottmv{x}$}

    \\
    %% \gamma \, \ottnt{p} \, e & = & e  \ottsym{(}  \ottnt{p'}  \ottsym{)}
    %%                       & \mbox{when $ \ottnt{P}  [  \ottnt{p}  ]  = \ottmv{x}$,}\\
    %%                 & & &  \mbox{\quad $\ottmv{x}$ bound by $ \ottnt{P}  [  \ottnt{p'}  ]  =  \lambda \ottmv{x} . \underline{\ \ }  $ or} \\
    %%                 & & & \mbox{\quad by $ \ottnt{P}  [  \ottnt{p'}  ]  = \underline{\ \ } \, \kw{to} \, \ottmv{x} \, \kw{in} \, \underline{\ \ } $} 

    %% \\
    %% \gamma \, \ottnt{p} \, e & = & \ottsym{[}   \ottsym{0}  \CONS  \ottnt{p'}   \ottsym{,}  e  \ottsym{]}
    %%                     & \mbox{when $ \ottnt{P}  [  \ottnt{p}  ]  = \ottmv{x}$, $\ottmv{x}$ bound} \\
    %%                     &  & & \mbox{\quad by $ \ottnt{P}  [  \ottnt{p'}  ]  = \texttt{\textcolor{red}{<<no parses (char 1): r***ec x. \_ >>}}$} 

    %% \\                            
    \overline{\gamma} \, \ottnt{p} & = & \kwl{LBL} \, \ottsym{(}   \eta\;   \ottsym{0}  \CONS  \ottnt{p}    \ottsym{)} 
                        & \mbox{when $ \ottnt{P}  [  \ottnt{p}  ]  = \kw{thunk} \, \underline{\ \ }$}

    \\
     \overline{\gamma} \, \ottnt{p} & = & \ottmv{n} 
                         & \mbox{when $ \ottnt{P}  [  \ottnt{p}  ]  = \ottmv{n}$} 
\end{array}
\]

  \caption{Lock-step compilation strategy for the term
    $\ottnt{P}$.}
  \label{fig:compile}
\end{figure*}

The compiler is shown in Figure~\ref{fig:compile}. For each subterm, the compiler records as much information about its execution
as can be computed statically. The control constructs, sequencing, application, and $ \kw{letrec} $,
are not compiled at all, since control flow is encoded in the CFG edges.
For abstractions and $ \kw{prd} $ when the
corresponding argument frame is present, the path to bind in the
environment, the next instruction, and the operand computed from the value
being bound are recorded. Otherwise, for a lambda we can only record
the next instruction, and for a $ \kw{prd} $, the computed operand.
Arithmetic operators follow a similar pattern. For the $ \kw{force} $
instruction, we record the argument stack up to the nearest $ \kwl{SEQ} $
frame. If one is found, we record the the path to bind and the path to
the next instruction.

%% CFG

\paragraph{Correctness.}
CFG machine states are the same as PEK machine states, so there is no need for unloading and the well-formedness conditions are unchanged.

\begin{lemma} (CFG Machine Correctness) For all terms $\ottnt{P}$ and CFG
  states $s$ that are well-formed with respect to $\ottnt{P}$, 
\[
\begin{array}{c}
\kw{PEK.step} \, \ottnt{P} \, s =
 \kw{CFG.step} \, \ottsym{(}  \kw{compile} \, \ottnt{P}  \ottsym{)} \, s
\end{array}
\]
\end{lemma}

%\subsection{Modified compilation strategy.}
%\label{sec:cfg1}
%\steve{Cut this?}
%The steps of the CFG machine we have derived correspond exactly to
%those of the CEK machine. However, recall from Section~\ref{sec:proof} that steps of
%the CEK machine do not correspond exactly to steps of CBPV terms. In
%particular, the search rules in the SOS for CBPV are applied
%recursively, while each application requires a step in the CEK
%machine. These correspond to $ \kwl{NOP} $ instructions in the compiled
%program.  Fortunately, we can eliminate these
%steps without needing a new virtual machine: we only need a new
%compilation strategy. This improved compilation strategy will never
%emit instructions that jump to structural paths, i.e., those that
%point to sequences, applications, and projections. Instead, the
%instructions emitted will always skip to the location of the next
%non-trivial instruction. We denote the improved compiler as $\texttt{\textcolor{red}{<<no parses (char 8): compile1*** >>}}$. 

%\paragraph{Unloading, well-formedness of states, and correctness for
%  the improved compilation strategy.} 
%\begin{lemma} (Compilation Strategy Correctness) For all terms $P$ and CFG
%  states $s$ that are well-formed with respect to $P$, 
%  \begin{align*}
%    \mathrm{step_{CFG}^n}(compile(P), \mathrm{unload_{CFG1}} (s)) = \\
%    \mathrm{unload_{CFG1}} (\mathrm{step_{CFG}} (compile(P), s))
%  \end{align*}
%\end{lemma}
%
We have now seen how to translate from CBPV terms to CFG programs
through a series of machines, each one of which comes closer to an
assembly-language-like execution model. We define the top-level \kw{unload} function
as the composition of each of the unloading functions, translating a CFG state
under the improved compilation strategy to a CBPV term. At each level, we have proved a
lemma stating that the transformation preserves the semantics of
machine states. As described in Section~\ref{sec:proof}, we can combine the lemmas
at each level to obtain correctness properties for the entire
translation from CBPV to CFG.

%% Local Variables:
%% fill-column: 70
%% eval: (auto-fill-mode)
%% eval: (flyspell-mode)
%% eval: (setq-local auto-hscroll-mode nil)
%% eval: (setq-local sentence-end-double-space nil)
%% End:

%% SAZ: Deleted this section
% \section {Controlling Code Generation}
% \label{sec:codegen}
% \input{temp/codegen}

\section {Applications}
\label{sec:applications}
%\steve{We need to expand this section with more detailed discussion of
%  porting equivalences from CBPV to CFG.  It should also include
%  example optimizations.}
%  
Our precise connection between Levy's call-by-push-value calculus and
SSA-form control flow graphs has useful applications in compiler theory and in practice. 
On the theoretical side, it precisely relates %a step towards precisely relating 
work on imperative and on functional compilation. Moreover, for compiling functional languages, 
it suggests a new compilation path using call-by-push-value as an intermediate representation.

On the practical side, our connection reduces the cost of verifying a wide class of  SSA-form CFGs optimizations. 
%~\cite{haskellpaper,zenapaper} 
As we will demonstrate in Section~\ref{ssec:opt:examples}, reasoning about optimizations directly on SSA-form CFGs 
is expensive. Just as functional 
languages increase programmer productivity, they also increase verification productivity.  

We have developed a sound equational theory for our CBPV operational semantics.
Many optimizations can be expressed as rewrites in context and justified by
the equational theory. % of the associated calculus. 
For example, beta reduction 
can express transformations like constant folding, function inlining, and dead
code elimination~\cite{appeljim}.
While our low-level execution model does not enable new optimizations
on control flow graphs, it dramatically reduces the cost of verifying %makes it possible to directly express
transformations that could previously only be reasoned about directly
using CFG machine semantics. For example, Fig. \ref{fig:opt} shows a
CBPV program that computes the sum $(2 + 2) + (3 + 3)$ using an
auxiliary function that doubles its argument. In the first version of
the program, the function $f$ is bound using a lambda abstraction; in
the second, it is bound using $ \kw{letrec} $. The two programs reduce to a
common term, so the transformation can be justified by the CBPV
equational theory. The compiled CFG programs differ only in that the
auxiliary function is called directly in the optimized program.

Note that this optimization cannot be expressed simply as carrying out
the application in the original program. Substituting for $f$ would
duplicate the function body, which is not always desirable. The
ability to express a direct call crucially depends on having two types
of binders with different execution behaviors.

The rest of this section contains a summary of our CBPV equational theory as well as a demonstration of its value in verifying typical examples of compiler optimizations. 
Once the sound equational theory was developed, verifying each of the optimizations was entirely straightforward and took less than an hour. We demonstrate the corresponding CFG transformations produced by such optimizations to illustrate how difficult reasoning about such optimizations would be had we reasoned directly on the CFGs.

\lstdefinelanguage{CFG}{%
  morekeywords={%
    POP,MOV,TAIL,IF0,RET,SUB,TAIL,ADD % keywords go here
  },%
  morecomment=[l]{//},
  otherkeywords={:,@}, % operators go here
  basicstyle={\small\ttfamily},
  keywordstyle={\bfseries},
  keywordstyle={[2]\itshape}, % style for types
  commentstyle={\rmfamily\itshape},
  keepspaces,
  mathescape % optional
}[keywords,comments,strings]%

 %% (0, (["TAIL"; "@1"], [])); 
 %%
 %% (2, (["POP"], [3]));
 %% (3, (["%5 = ADD %2 %2"], [4])); 
 %% (4, (["RET %5"], []));
 %%
 %% (1, (["PMOV"; "%1 = @2"], [6])); 
 %% (6, (["%8 = CALL %1"; "$2"], [7]));
 %% (7, (["%10 = CALL %1"; "$3"], [9]));
 %% (9, (["%12 = ADD %8 %10"], [11])); 
 %% (11, (["RET %12"], []))

 %% (0, (["TAIL"; "@1"], [])); 
 %%
 %% (3, (["POP"], [10]));
 %% (10, (["%12 = ADD %3 %3"], [11])); 
 %% (11, (["RET %12"], []))]
 %%
 %% (1, (["%5 = CALL @3"; "$2"], [4]));
 %% (4, (["%7 = CALL @3"; "$3"], [6])); 
 %% (6, (["%9 = ADD %5 %7"], [8]));
 %% (8, (["RET %9"], []));

\begin{figure}[tp]
%% Example knownfn : tm :=
%% Force (Rec (App
%%               (Rec (Lam (Seq (Aop OPlus (Var 0) (Var 0))
%%                              (Prd (Var 0)))))
%%               (Lam (Seq (App (Nat 2) (Force (Var 0)))
%%                    (Seq (App (Nat 3) (Force (Var 1)))
%%                    (Seq (Aop OPlus (Var 1) (Var 0))
%%                         (Prd (Var 0)))))))).
\begin{minipage}[t]{0.5\textwidth}
\[
\begin{array}[t]{l}
 \kw{thunk} \ ( \kw{force} \ ( \kw{thunk} \ (\\
\quad ( \kw{force} \ ( \kw{thunk} \ (a \oplus b\ \kw{to}\ y\ \kw{in}\ \kw{prd} \, \ottmv{y})))\ \kw{to}\ x\\ \quad\kw{in}\ \kw{prd} \, \ottmv{x})))\end{array}
\]
\end{minipage}
\begin{minipage}[t]{0.2\textwidth}
\begin{lstlisting}[language=CFG]
0:  TAIL @1

1:  x = CALL @3 [2]
2:  RET x

3:  y = ADD a b [4]
4:  RET y
\end{lstlisting}
\end{minipage}

\begin{minipage}[t]{.5\textwidth}
\[
\begin{array}[t]{l}
 \kw{thunk} \ ((a \oplus b\ \kw{to}\ y\ \kw{in}\ \kw{prd} \, \ottmv{y})\\
\qquad  \kw{to}\ x\ \kw{in}\ \kw{prd} \, \ottmv{x})
\end{array}
\]
\end{minipage}
\begin{minipage}[t]{0.2\textwidth}
\begin{lstlisting}[language=CFG]
0:  y = ADD a b [1]
1:  MOV y x [2]
2:  RET x
\end{lstlisting}
\end{minipage}

\begin{minipage}[t]{.5\textwidth}
\[
\begin{array}[t]{l}
 \kw{thunk} \ (a \oplus b)
\end{array}
\]
\end{minipage}
\begin{minipage}[t]{0.2\textwidth}
\begin{lstlisting}[language=CFG]
0:  ORET ADD a b
\end{lstlisting}
\end{minipage}

\caption{Three equivalent CBPV terms after successive optimizations
  (left) and their corresponding CFG representations.  The first
  optimization pass inlines a function call and merges blocks (by
  eliminating $ \kw{force} $--$ \kw{thunk} $ pairs).  The second optimization
  pass removes the $ \kwl{MOV} $ by eliminating the unnecessary
  $\kw{prd} \, \ottmv{x}$.
}
\label{fig:opt}
\end{figure}

%% \begin{minipage}[t]{.33\textwidth}
%% \[
%% \begin{array}[t]{l}
%% \kw{power} =\\
%% \lambda n. \lambda m.\\
%% \quad 1  \mathbin{\cdot}  m  \mathbin{\cdot}  
%% \overline{\kw{prj}_0} \mathbin{\cdot}   \kw{force}  (\texttt{\textcolor{red}{<<no parses (char 1): r***ec >>}}\; x.\\
%% \\
%% \qquad \qquad  \langle  \lambda m_1. \lambda a_1.\\
%% \qquad \qquad \qquad \kw{if0}\;m_1\\
%% \\
%% \qquad \qquad \qquad \quad ( \kw{prd} \; a_1)\\
%% \\
%% \qquad \qquad \qquad \quad ((m_1 - 1)\; \kw{to} \;m_2\; \kw{in} \\
%% \qquad \qquad \qquad \qquad \quad m_2  \mathbin{\cdot}  0  \mathbin{\cdot}  a_1  \mathbin{\cdot} 
%%   \overline{\kw{prj}_1} \mathbin{\cdot}  ( \kw{force} \;x))\; ,\\
%% \\
%% \qquad \qquad \  \lambda m_3. \lambda a_2. \lambda m_k.\\
%% \qquad \qquad \qquad \kw{if0}\;m_3\;\\
%% \\
%% \qquad \qquad \qquad \quad
%%   (a_2 \mathbin{\cdot} m_k \mathbin{\cdot} \overline{\kw{prj}_0} \mathbin{\cdot} ( \kw{force} \;x)) \\
%% \\
%% \qquad \qquad \qquad \quad ((m_3 - 1)\; \kw{to} \;m_4\; \kw{in} \\
%% \qquad \qquad \qquad \quad (a_2 + n)\; \kw{to} \;a_3\; \kw{in} \\
%% \qquad \qquad \qquad \qquad \quad \quad
%%   m_k \mathbin{\cdot} a_3 \mathbin{\cdot} m_4 \mathbin{\cdot} \overline{\kw{prj}_1} \mathbin{\cdot} ( \kw{force} \;x))\\
%% \qquad \qquad  \rangle )\\
%% \end{array}
%% \]
%% \end{minipage}

\lstdefinelanguage{CFG}{%
  morekeywords={%
    POP,MOV,TAIL,IF0,RET,SUB,TAIL,ADD % keywords go here
  },%
  morecomment=[l]{//},
  otherkeywords={:,@}, % operators go here
  basicstyle={\small\ttfamily},
  keywordstyle={\bfseries},
  keywordstyle={[2]\itshape}, % style for types
  commentstyle={\rmfamily\itshape},
  keepspaces,
  mathescape % optional
}[keywords,comments,strings]%

 %% (0, (["TAIL"; "@1"], [])); 
 %%
 %% (2, (["POP"], [3]));
 %% (3, (["%5 = ADD %2 %2"], [4])); 
 %% (4, (["RET %5"], []));
 %%
 %% (1, (["PMOV"; "%1 = @2"], [6])); 
 %% (6, (["%8 = CALL %1"; "$2"], [7]));
 %% (7, (["%10 = CALL %1"; "$3"], [9]));
 %% (9, (["%12 = ADD %8 %10"], [11])); 
 %% (11, (["RET %12"], []))

 %% (0, (["TAIL"; "@1"], [])); 
 %%
 %% (3, (["POP"], [10]));
 %% (10, (["%12 = ADD %3 %3"], [11])); 
 %% (11, (["RET %12"], []))]
 %%
 %% (1, (["%5 = CALL @3"; "$2"], [4]));
 %% (4, (["%7 = CALL @3"; "$3"], [6])); 
 %% (6, (["%9 = ADD %5 %7"], [8]));
 %% (8, (["RET %9"], []));

\begin{figure*}[htbp]
%% Example knownfn : tm :=
%% Force (Rec (App
%%               (Rec (Lam (Seq (Aop OPlus (Var 0) (Var 0))
%%                              (Prd (Var 0)))))
%%               (Lam (Seq (App (Nat 2) (Force (Var 0)))
%%                    (Seq (App (Nat 3) (Force (Var 1)))
%%                    (Seq (Aop OPlus (Var 1) (Var 0))
%%                         (Prd (Var 0)))))))).
\begin{minipage}[t]{.5\textwidth}
\[
\begin{array}[t]{l}
 \kw{letrec} \ x_1 = \lambda a.\ a + a\ \kw{to}\ x\ \kw{in}\ \kw{prd} \, \ottmv{x},\\
\quad \quad x_2 =  2 \cdot  \kw{force} \ x_1\ \kw{to}\ x\\
\quad\quad\quad\kw{in}\ 3 \cdot  \kw{force} \ x_1\ \kw{to}\ y\\
\quad\quad\quad\kw{in}\ x + y\ \kw{to}\ z\\
\quad\quad\quad\kw{in}\  \kw{prd} \ z\\
\kw{in}\  \kw{force} \ x_2
\end{array}
\]
\end{minipage}
\begin{minipage}[t]{0.2\textwidth}
\begin{lstlisting}[language=CFG]
0:  TAIL @1

1:  a = POP [2]
2:  w = ADD a a [3]
3:  RET w

4:  x = CALL @1 2 [5]
5:  y = CALL @1 3 [6]
6:  z = x + y [7]
7:  RET z
\end{lstlisting}
\end{minipage}

\begin{minipage}[t]{.5\textwidth}
\[
\begin{array}[t]{l}
 \kw{force} \ ( \kw{thunk} \ (\\
\quad  \kw{thunk} \ (\lambda a.\ a + a\ \kw{to}\ x\ \kw{in}\ \kw{prd} \, \ottmv{x})\ \cdot\\
\quad \lambda f.\ 2 \cdot  \kw{force} \ f\ \kw{to}\ x\\
\quad\quad\kw{in}\ 3 \cdot  \kw{force} \ f\ \kw{to}\ y\\
\quad\quad\kw{in}\ x + y\ \kw{to}\ z\\
\quad\quad\kw{in}\  \kw{prd} \ z))
\end{array}
\]
\end{minipage}
\begin{minipage}[t]{0.2\textwidth}
\begin{lstlisting}[language=CFG]
0:  TAIL @4
1:  a = POP [2]
2:  w = ADD a a [3]
3:  RET w

4:  MOV @1 f [5]
5:  x = CALL f 3 [6]
6:  y = CALL f 3 [6]
6:  z = x + y [7]
7:  RET z
\end{lstlisting}
\end{minipage}

\caption{
  An optimization replacing an indirect call to $f$ with a direct call.
  \textbf{Left}: The original and optimized CBPV program. \textbf{Right}:
  The resulting CFG programs.}
\label{fig:opt2}
\end{figure*}

\subsection{Equational Theory}
 
 Compilers rely on program equivalences when optimizing and transforming programs.
In the context of verified compilation, as found in the
CompCert~\cite{compcert} and Cake ML~\cite{cakeml} projects, formal verification of
particular program equivalences is crucial: the correctness of classic
optimizations like constant folding, code inlining, loop unrolling,
etc., hinges on such proofs.  Techniques that facilitate
such machine-checked proofs therefore have the potential for broad
impact in the domain of formal verification.

Establishing program equivalence is a well-known and long-studied problem. %~\cite{} .
The gold standard of what it means for two (potentially open) program
terms $\ottnt{M_{{\mathrm{1}}}}$ and $\ottnt{M_{{\mathrm{2}}}}$ to be equal is \textit{contextual
  equivalence}.  Although the exact technical definition varies slightly 
from language to language (and we formalized our precise
definition in Coq), the intuition is that $\ottnt{M_{{\mathrm{1}}}}$ is contextually
equivalent to $\ottnt{M_{{\mathrm{2}}}}$ if $C[ \ottnt{M_{{\mathrm{1}}}} ]$ ``behaves the same'' as $C[ \ottnt{M_{{\mathrm{2}}}} ]$
for every closing context $C[-]$.  Here, ``behaves the same'' can simply be
that both programs terminate or both diverge.  Such contextual
equivalences fully justify program optimizations where we can think of
replacing a less-optimal program $\ottnt{M_{{\mathrm{1}}}}$ by a better $\ottnt{M_{{\mathrm{2}}}}$ in the
program context $C[-]$, without affecting the intended behavior of the
program. 

While the programming languages literature is filled with powerful
general techniques for establishing program equivalences using
pen-and-paper proofs, few of these general techniques have been 
mechanically verified in a theorem prover.
%Moreover, some of them are quite difficult to apply in practice.  For
%example, complete methods, such as applicative
%bisimulation~\cite{applicative} and ``closed-instances of uses'' (CIU)
%techniques~\cite{CIU}, require quantification over \textit{all} closed
%function arguments or closing contexts (respectively), which
%significantly complicates the proofs.

Lassen introduced the notion of
\textit{eager normal form bisimulations}~\cite{StovringL09}, which is applicable
to the call-by-value lambda calculus and yet avoids quantification
over all arguments or contexts.  Such normal form simulation is
\textit{sound} with respect to contextual equivalence but it is not
\textit{complete} (this is the price to be paid for the
easier-to-establish equivalences); however, it is adequate for many
equivalence proofs because the relation contains reduction and is a
congruence.

Lassen's approach is appealing for formal verification because it
relies on simple co-inductive relations defined in terms of the structural
operational semantics.  Inspired by the this potential, 
we developed a proof technique based on normal form bisimulations for CBPV and
formalized its meta-theory in Coq.

%The broadest notion of program equivalence is \emph{contextual equivalence} which equates two terms if their behavior is the same in all program contexts. The universal quantification over contexts makes proving program equivalence using the definition directly impractical~\cite{StovringL09}. 

%Similar to our prior trick where we inductively define parallel compatible closure to get around the universal quantification over contexts when defining the equational theory, we also inductively define the \emph{compatible closure} over terms in order to avoid the universal quantification in the definition of contextual equivalence. 
%In fact, as we are only interested in wellformed terms we define a \emph{conditional compatible closure} which restricts the context to that of terms which respect the condition. The condition is later used to restrict the context to wellformed terms.

The soundness theorem states that our equational theory is included in contextual equivalence. 
The proof of this theorem is involved and it is not the focus of this paper; it will be explained in detail in a separate note. 
In summary, the proof follows directly from three main theorems:
 The first theorem states that the equational theory is a compatible closure on wellformed terms, meaning that if two wellformed terms are related by the equational theory then so are terms in their contextual closure. 
 The second theorem states that wellformed terms related by our normal form bisimulation relation co-terminate. 
 Our definition of normal form bisimulation is specifically designed to simplify this proof. 
 The third and most involved theorem states that wellformed terms that are related by the equational theory are also related by the normal form bisimulation relation.  
 
 \subsection{Examples}
 \label{ssec:opt:examples}

The following table summarizes various desired CFG compiler optimizations that are easily proven correct using our CBPV equational theory by term rewriting. 

 \vspace{1em}
 \begin{tabular}{l|l}
 \label{tab:opts}
 CBPV Equation & CFG Optimization\\
\hline\\
$\kw{force} \, \ottsym{(}  \kw{thunk} \, \ottnt{M}  \ottsym{)} = \ottnt{M}$ 	& block merging   ``direct jump case" \\[0.5ex]
%
%(App v (Lam x. M)) = {v/x}M
$\ottnt{V} \!\cdot\! \lambda \ottmv{x} . \ottnt{M}  = \ottsym{\{}  \ottnt{V}  \ottsym{/}  \ottmv{x}  \ottsym{\}} \, \ottnt{M}$ &block merging   ``phi case" \\[0.5ex]

%Seq M to x in (Prd x) = M  
$\ottnt{\kw{prd} \, \ottnt{V}} \, \kw{to} \, \ottmv{x} \, \kw{in} \, \ottnt{M} = \ottsym{\{}  \ottnt{V}  \ottsym{/}  \ottmv{x}  \ottsym{\}} \, \ottnt{M}$ & move elimination \\[0.5ex]
%Seq (Aop op N1 N2) M = subst 0 (N1 \texttt{\textcolor{red}{<<no parses (char 2): op*** >>}} N1) M 
$\ottnt{\ottsym{(}   n_1   \oplus   n_2   \ottsym{)}} \, \kw{to} \, \ottmv{x} \, \kw{in} \, \ottnt{M} = \ottsym{\{}   n_1   \means{\oplus}   n_2   \ottsym{/}  \ottmv{x}  \ottsym{\}} \, \ottnt{M}$ &   constant folding \\[0.5ex]
 %
%
% App (Thunk (Lam y.M)) (Lam x.N) = {Thunk(Lam M)/x} N 
${\kw{thunk} \ (\lambda {y} . \ottnt{M})} \cdot\! \lambda \ottmv{x} . \ottnt{N}  = \ottsym{\{}  { \kw{thunk} \ (\lambda {y} . \ottnt{M})}   \ottsym{/}  \ottmv{x}  \ottsym{\}} \, \ottnt{N}$ & function inlining\\[0.5ex]
%
%IF True M N = M  \textbf{and}	IF False M N = N	&  		Dead Branch Elimination \\
 $\kw{if0} \;  \ottsym{0} \;  \ottnt{M_{{\mathrm{1}}}} \;  \ottnt{M_{{\mathrm{2}}}}  = \ottnt{M_{{\mathrm{1}}}}$  & dead branch elimination ``true branch"\\
 $\kw{if0} \;   n  \;  \ottnt{M_{{\mathrm{1}}}} \;  \ottnt{M_{{\mathrm{2}}}}  = \ottnt{M_{{\mathrm{2}}}}$ where \ $( n \not = 0)$ &dead branch elimination  ``false branch" \\[0.5ex]
% 
%
%IF B M M = M	&          Conditional Execution (or Branch Elimination) \\
$\kw{if0} \;  \ottnt{V} \;  \ottnt{M} \;  \ottnt{M}  = \ottnt{M}$ & branch elimination

%IFZ M (C[IFZ M M1 M2)) (D[IFZ M N1 N2))   =    IFZ M C[M1] D[N2]   (where FV(M) not bound in C or D)    & \\
%Rec C[Force x1] (x = M ..) = Rec C[M] (x1 = M ..)        when FV(M) = x_i    &   function inlining? \\
%Rec C[Force x1] (x = M ..) = Rec C[M] (x1 = M ..)        in general               &     also loop unrolling? \\
 \end{tabular}
\vspace{1em}

By relying on our CBPV equational theory it was simple to prove all these optimizations correct in Coq with a minimal amount of effort. This demonstrates the power of our approach in reducing the cost of verifying compiler optimizations. 
The same approach will similarly facilitate reasoning about more complex optimizations.
  
  %In the rest of the section, we describe a few of the examples from the previous table in more detail.
 
 \paragraph{Constant folding for arithmetic operations:} 
 replaces the plus operation on numbers by its resulting value, and leaves the rest of the program unchanged. 
%More precisely, given a program it substitutes all sub-programs of the form $ n_1   \oplus   n_2  \, \kw{to} \, \ottmv{x} \, \kw{in} \, \ottnt{N}$ by $\ottsym{\{}   n_1   \means{\oplus}   n_2   \ottsym{/}  \ottmv{x}  \ottsym{\}} \, \ottnt{N}$, where 
%$ \oplus  = +$. 
 %
 
\noindent
\paragraph{Function inlining: } replaces all function calls in the program with the bodies of the called functions. 
 
 \noindent
\paragraph{Dead code elimination for conditionals: } removes dead branches in conditional statements.

\vspace{\baselineskip} 
Verifying these optimizations was trivial using our CBPV equational theory because the optimization is just a step in our CBPV semantics, and because our equational theory is sound with respect to the CBPV semantics and equates beta-equivalent terms.
 
 \noindent
\paragraph{Identical branch simplification: } replaces conditional statements that have identical branches with the code of one of the branches. Unlike the three previous optimizations, this one does not follow directly from one step in the operational semantics. 
  However, it was still quite easy to prove its correctness.

\vspace{\baselineskip}
  
Figure~\ref{fig:opt} contains a CBPV program and a program that is optimized using several of these optimizations. 
It also contains the CFGs of the original and the optimized code. As we can see, the relation between the two CFG graphs is non-trivial.

\section{Related Work}
\label{sec:discussion}
The similarity between SSA intermediate representations and functional
languages using lexical scope was noticed early on by Appel~\cite{appel1992}.
Kelsey provided translations between the
intraprocedural fragment of SSA and a subset of CPS with
additional annotations~\cite{kelseyssa}. Chakravarty et al.~\cite{ChakravartyKZ03} extended
the correspondence to ANF forms, and presented a translation from SSA
to ANF. They used this translation to port a variant of the sparse
conditional constant propagation algorithm from SSA to ANF, and proved
the associated analysis sound.
Real compilers using ANF or CPS representations employ sophisticated
analyses to avoid performing the extra work of saving and restoring
the environment, and to turn the calls into direct
jumps~\cite{wandselective,reppylocalcps,fluetcontification}. These
intensional aspects of execution, however, are outside of the scope of
the cost model provided by the ANF machine.

Beringer et al.~\cite{BeringerMS03} define a first-order language
%with mutually recursive but non-nested local functions
that has both ``functional'' and ``imperative'' semantics. They present a
big-step environment-based semantics where local functions are
represented as closures in the environment, and a small-step semantics
where variable bindings are interpreted as assignments and local
function calls as jumps. They then identify a restricted subset of
programs for which the two evaluation functions agree. The approach is
similarly motivated by the more powerful reasoning principles
available for the ``functional'' semantics.

Recent work by by Schneider et al.~\cite{SchneiderSH15} uses a similar
approach, but their first-order language allows nested local function
declarations as well as observable effects. The subset of
terms for which the functional and imperative semantics of their
language agree includes functions with free variables bound in an
enclosing scope, using a condition they call \textit{coherence}. They
then formulate a register allocation transformation that puts programs
in coherent form, and prove that it preserves a form of trace
equivalence. They also provide machine-checked proofs.

The approach of linking existing calculi and virtual machines was
inspired by the methods of Ager et
al.~\cite{danvymachines,agervirtual}. The method of deriving compilers
and virtual machine interpreters is attributed to Wand~\cite{Wand85}.
Many authors have investigated the relationships between various forms
of operational semantics~\cite{hannan90,curien91}, and it continues
to be an active topic of research~\cite{danvy04,danvy08}. Our
contribution can be seen as an extension of this line of work to
control flow graph formalisms found in the compiler literature.

Compared to the derivations presented in~\cite{danvyvirtual}, our target language (and
therefore our approach) is more ad-hoc: we do not make any
connections to existing functional transformations such as
defunctionalization or CPS. The resulting
compiler also does more work by emitting
different instructions based on context. Rather than
beginning the derivation with an evaluation function, we use
small-step operational semantics throughout.

% \subsection{Conclusions}
% Our approach to the correspondence between SSA form and functional
% programming is to recast it as the relationship between an abstract
% machine and a source calculus. We have extended existing work in several
% key ways:
% \begin{itemize}
% \item Our source language is an existing higher-order lambda calculus
% with well-studied semantics and no additional syntactic restrictions
% or extensions.
% \item Our target is a machine operating on an explicit CFG
% representation of programs, similar to intermediate representations
% used by existing C compilers and recent work on verified language
% implementations~\cite{compcert, vellvm}.
% \item We prove a stronger correctness property: Our functional and
% imperative semantics coincide for all source calculus terms, rather
% than a carefully chosen subset. We also show that the structural
% operational semantics and virtual machine implement the same
% transition system (machine equivalence) instead of using a fixed notion of observational
% equivalence.
% \item We show how the source calculus can be equipped with ad-hoc
% annotations to extend the equivalence to additional instructions in
% the CFG machine, without compromising SOS-based reasoning.
% \end{itemize}

%% Local Variables:
%% fill-column: 70
%% eval: (auto-fill-mode)
%% eval: (flyspell-mode)
%% eval: (setq-local auto-hscroll-mode nil)
%% eval: (setq-local sentence-end-double-space nil)
%% End:

\section{Conclusion}
\label{sec:conclusion}
%% \begin{itemize}
%% \item Why not let-normalize? \\
%%   * ANF / CPS achieve a simple execution model by first translating to \\
%%   terms with greatly limited eval contexts \\
%%   * Wouldn't this simplify the PEAK machine? \\
%%   * PEAK does a kind of let-normalization during execution \\
%%   * Generated code would look similar \\
%%   * BUT ANF and CPS introduce additional redexes: correctness is often
%%   stated as plus simulation \\
%%   * Introduce join point functions for conditionals \\
%%   * Breaks simple lock-step machine property with source term: have to
%%   reason about result of transformation \\
%%   * This is an alternative that skips that particular step of
%%   compilation \\
%%   * Also, ANF is not closed under beta \\
%%   * It's unclear what CPS on CBPV would accomplish: no CBN/CBV \\
%%   * EoCC shows that an efficient execution strategy for CPS terms is
%%   highly nontrivial \\
%% \item CFG representations of lambda terms \\
%%   * Appel/Jim, Kennedy, different goals entirely: \\
%%   * Goal is efficient datastructure for optimization \\
%%   * Graph structure not intended to be execution model \\
%%   * SSA is both a datastructure and machine
%% \item Alternative CFG choices \\
%%   * Our goal was to replicate existing SSA IRs \\
%%   * Definitely a lot of choices for how arguments are used \\
%%   * E.g. separate push instruction, batch pop on return \\
%%   * Arith operations are sequenced even though they have no effects \\
%%   * Could treat these as values vs. computations \\
%%   * Don't force sequencing: sea-of-nodes IR? Click
%% \end{itemize}

Operational semantics and equational theories are defined on lambda calculi; compiler optimizations are applied to control flow graphs, often in SSA form. In this paper, we have shown that the two are fundamentally the same object: a call-by-push-value term can be translated to an SSA CFG such that a step of reduction in the term corresponds exactly to a step of execution in the CFG. We have proved this machine equivalence by using a chain of intermediate machines, gradually separating out the statically computable aspects of reduction until we arrive at a compiler and a CFG execution model. The tight correspondence between CBPV and CFG semantics gives us a new tool for verifying CFG-level program transformations, by proving equivalence of the corresponding CBPV terms in an equational theory.

%% Local Variables:
%% fill-column: 70
%% eval: (auto-fill-mode)
%% eval: (flyspell-mode)
%% eval: (setq-local auto-hscroll-mode nil)
%% eval: (setq-local sentence-end-double-space nil)
%% End:

% \bibliographystyle{abbrvnat}
% \bibliography{bibliography}

%\nocite{*}
\bibliographystyle{abbrvnat}
{\small
\bibliography{biblio}
}
%% \newpage
%% \appendix
%% \section{Appendix: Collected Semantics}
%% \label{sec:appendix}
%% \input{temp/appendix}
% \clearpage
% \section{Notes}
% \input{temp/notes}
\end{document}